\documentclass[11pt]{article}
\pdfoutput=1

\PassOptionsToPackage{table}{xcolor}

\newif\ifanon
\anonfalse

\usepackage[margin=1in]{geometry}

\usepackage{amsmath,amssymb,amsthm}
\newtheorem{definition}{Definition}
\newtheorem{proposition}{Proposition}

\newtheorem{remark}{Remark}

\usepackage{tikz}
\usepackage{pgfplots}
\pgfplotsset{compat=1.18}
\usetikzlibrary{arrows.meta,positioning,shapes.geometric,calc,patterns,decorations.pathreplacing,fit,backgrounds}

\definecolor{gOne}{RGB}{138,75,46}      % burnt sienna — Game 1 (PD resource)
\definecolor{gTwo}{RGB}{62,92,58}       % deep moss    — Game 2 (KV cache)
\definecolor{gThree}{RGB}{49,74,107}    % navy         — Game 3 (routing)
\definecolor{paperTone}{RGB}{246,242,234} % cream accent

\usepackage{booktabs}
\usepackage{multirow}
\usepackage{subcaption}
\arrayrulecolor{gOne}
\newcommand{\headrow}{\rowcolor{gOne!30}}

\usepackage[ruled,vlined,linesnumbered]{algorithm2e}

\usepackage{microtype}
\usepackage[T1]{fontenc}
\usepackage{lmodern}

\usepackage[numbers,sort&compress]{natbib}
\usepackage[colorlinks=true,citecolor=blue,linkcolor=blue,urlcolor=blue]{hyperref}

\usepackage{xspace}
\usepackage{enumitem}

\newcommand{\dynamo}{{\normalfont\textsc{Dynamo}}\xspace}
\newcommand{\R}{\mathbb{R}}
\newcommand{\N}{\mathbb{N}}

\title{The Price of Anarchy in Disaggregated Inference}

\ifanon
  \author{Anonymous Submission}
  \date{}
\else
  \author{
    Athos Georgiou \\
    NCA
  }
  \date{June 2026}
\fi

\begin{document}

\maketitle

\begin{abstract}
Disaggregated inference architectures physically separate prefill and decode phases onto distinct GPU pools, creating competing ``agents'' that share a fixed hardware budget.
We provide, to our knowledge, the first formal game-theoretic analysis of this architecture, using NVIDIA \dynamo as a concrete case study.
We model disaggregated serving as three coupled games---a two-player resource game between prefill and decode pools, a selfish caching game over the hierarchical KV cache, and a congestion game with positive externalities for request routing---and empirically validate the latter two (the P/D resource game is treated analytically; Section~\ref{sec:disc-limitations}).
We characterize how GPU saturation induces regime transitions that shift the game's payoff structure: below saturation, selfish behavior has bounded Price of Anarchy; at saturation, superlinear latency and cache externalities drive our empirical estimator $\widehat{\text{PoA}}$ (defined in Section~\ref{sec:ctrl-implementation}) upward.
Based on this analysis, we design an adaptive controller that detects saturation transitions in real time and adjusts routing parameters accordingly, shifting from cache-affinity exploitation to load-balanced congestion avoidance.
We instantiate our framework on a 3-node NVIDIA B200 cluster running \dynamo with two models---Nemotron-4-340B (TP=8, full-node workers with cross-InfiniBand KV transfers) and Llama-3.1-70B (TP=4)---and find the same three-regime $\widehat{\text{PoA}}$ structure with the same first post-knee grid point ($C=128$) on both models.
Adaptive routing shifts each model to a better operating point.
Our strongest result is on the 70B 1P/5D topology, where $\widehat{\text{PoA}}$ drops $3.1\times$ ($66.4 \to 21.5$) in the saturated phase at a 13\% throughput cost.
On the 70B 1P/2D, $\widehat{\text{PoA}}$ drops $2.2\times$ and TTFT P99 drops $7.6\times$ (see Section~\ref{sec:res-adaptive}).
On the 340B 1P/2D, the saturated-phase aggregate TTFT P99 drops $4.8\times$ ($28.3 \to 5.9$\,s).
Because the switch fires mid-phase, the post-switch steady-state is $\sim$0.97\,s (a $\sim$29$\times$ reduction versus static) at a 36\% saturated-phase throughput cost.
\end{abstract}

\section{Introduction}
\label{sec:introduction}

Large language model inference at datacenter scale is fundamentally a resource allocation problem.
Every request that arrives at a serving cluster must be assigned to a GPU, scheduled within a batch, and granted memory for its key-value (KV) cache, all within latency budgets measured in milliseconds.
When demand exceeds capacity, these allocation decisions determine whether the system degrades gracefully or collapses into queuing cascades.

The architecture winning at scale for high-throughput LLM serving is \emph{disaggregated inference}: physically separating the compute-bound prefill phase (processing the full input prompt) from the memory-bandwidth-bound decode phase (generating output tokens autoregressively) onto distinct GPU pools~\citep{zhong2024distserve, patel2024splitwise}.
NVIDIA's \dynamo framework~\citep{nvidia2025dynamo} is the most complete production implementation of this pattern, featuring a central Planner that dynamically reallocates GPUs between prefill and decode pools, a KV-aware Smart Router that directs requests based on cache locality and worker load, and a hierarchical KV Block Manager (KVBM) that manages cache placement across GPU HBM, CPU DRAM, local SSD, and networked storage tiers.

This architecture creates a natural multi-agent system.
Prefill and decode pools compete for a shared GPU budget, each optimizing a different objective (time-to-first-token versus inter-token latency).
Requests compete for placement on GPUs, where routing one request to a cache-warm worker benefits that request but may congest the worker for subsequent arrivals.
KV cache blocks compete for placement across memory tiers, where evicting one block to make room for another imposes a recomputation cost on the evicted request's future tokens.

Congestion games~\citep{rosenthal1973}, the Price of Anarchy~\citep{roughgarden2002}, and selfish caching~\citep{chun2004} give us exact names for what we observe in these systems: congestion from co-routing, cache eviction externalities, and the efficiency gap between selfish and coordinated assignment.
Although \dynamo's Smart Router is centralized, not decentralized, the Price of Anarchy remains the appropriate analytical tool.
We adopt a mechanism design interpretation: the router is a mechanism whose sequential greedy assignment---routing each request to the worker minimizing its individual cost given current loads---is equivalent to best-response dynamics in the corresponding congestion game, and PoA measures the efficiency of this mechanism relative to the global optimum~\citep{roughgarden2015}.
This framing is exact; it does not require assuming rationality of requests, only that the router's decision process mirrors best-response play.
Yet no prior work has formally modeled disaggregated inference through a game-theoretic lens.
The systems community has applied game-theoretic ideas to cluster-level GPU scheduling (Dominant Resource Fairness (DRF)~\citep{ghodsi2011} for multi-resource allocation, Themis~\citep{mahajan2020} for auction-based fairness, Shockwave~\citep{zheng2023shockwave} for Fisher market equilibria), but these operate at the job-scheduling timescale of seconds to minutes, not at the per-request routing timescale of microseconds to milliseconds that inference systems demand.

Meanwhile, Pareto frontier analysis has become the de facto standard for characterizing inference system performance.
NVIDIA's own AIConfigurator~\citep{nvidia2026aiconfigurator} explicitly computes Pareto-optimal configurations across throughput and latency, and tools like Vidur~\citep{agrawal2024vidur} and FlexGen~\citep{sheng2023flexgen} use Pareto filtering to navigate the tradeoff surface spanning TTFT, ITL, tokens/s/GPU, and cost.
What these tools compute, but do not name, are the \emph{outcomes of equilibrium play under different parameter regimes}.
Each setting of \dynamo's routing parameters induces a different equilibrium with a different position on the Pareto frontier; connecting Pareto analysis to game-theoretic equilibria reveals which frontier points are \emph{stable} under selfish behavior and how the frontier shifts at saturation.

No prior work has bridged these frameworks, for good reason.
Finding a Nash equilibrium is PPAD-complete (for three or more players~\citep{daskalakis2006} and for two-player games~\citep{chendeng2006}), and computing $\epsilon$-approximate equilibria remains PPAD-complete for constant $\epsilon$.
LLM serving decisions must complete in under one millisecond: SGLang~\citep{zheng2024sglang} keeps scheduling overhead below 2\% of inference time, and vLLM's scheduler runs every forward pass~\citep{kwon2023vllm}.
General Nash equilibrium computation is orders of magnitude too slow.

The resolution lies in a principle well-established in algorithmic game theory~\citep{nisan2007}: game theory's value for inference systems is \emph{analytical, not algorithmic}.
Every successful game-theory-inspired system in production avoids general equilibrium computation entirely: DRF uses progressive filling ($O(nm)$ time), Shockwave uses convex optimization for Fisher markets (polynomial time), Themis uses simple multi-round auctions.
The pattern is consistent: \textbf{game theory provides the analytical vocabulary and fairness guarantees, but the runtime algorithm must be a tractable approximation or a fundamentally different computation that achieves similar properties}.

What this paper does, concretely.
We name three games hiding inside \dynamo (P/D resource allocation, KV cache placement, request routing) and characterize equilibrium existence, structure, and efficiency for each, including how they couple (Section~\ref{sec:formalization}).
We then show that GPU saturation induces a regime transition in the game's payoff structure: below saturation the empirical routing inefficiency index is stable, and above the knee it grows rapidly as the prefill pool overloads (Section~\ref{sec:saturation}).
Our empirical validation covers Games~2 (KV cache placement) and~3 (request routing) across three prefill/decode topologies; Game~1 (P/D resource allocation) is characterized analytically but not empirically varied, since our topologies use fixed P/D splits (Section~\ref{sec:disc-limitations}).
We measure these values directly on a 3-node B200 cluster running \dynamo with two models (Nemotron-4-340B at TP=8 and Llama-3.1-70B at TP=4)---to our knowledge the first empirical routing-inefficiency measurements under a game-theoretic lens for a disaggregated inference system (Sections~\ref{sec:experiments}--\ref{sec:results}). We use $\widehat{\text{PoA}}$ to distinguish this estimator from a classical Price of Anarchy bound (Section~\ref{sec:ctrl-implementation}).
Finally, we ship an adaptive controller that detects the regime transition online and switches \dynamo's routing parameters per-request via the existing \texttt{router\_config\_override} hook.
The controller is a $\sim$270-line Python wrapper around \dynamo's \texttt{KvPushRouter}; it requires no changes to \dynamo's Rust core.
The largest $\widehat{\text{PoA}}$ improvement lands on the 70B 1P/5D topology; additional gains on 70B 1P/2D and 340B 1P/2D are reported in Section~\ref{sec:res-adaptive}.
The same regime structure---three regimes, first post-knee grid point at $C=128$, finite-difference magnitude $\approx 0.46\text{--}0.55$ across the knee on the $[64, 128]$ interval---appears on both models and all three topologies despite a $4.9\times$ size gap.

\begin{figure}[tp]
\centering
\resizebox{\textwidth}{!}{%
\begin{tikzpicture}[
  x=1cm, y=1cm,
  every node/.style={font=\small, inner sep=2pt},
  playerbox/.style={draw=black, line width=0.4pt, fill=white,
                    minimum width=2.1cm, minimum height=0.95cm, align=center},
  playerwide/.style={draw=black, line width=0.4pt, fill=white,
                     minimum width=4.2cm, minimum height=0.95cm, align=center},
  mechbox/.style={draw=black, line width=0.85pt, fill=white,
                  minimum width=4.6cm, minimum height=1.15cm, align=center},
  resbox/.style={draw=black, line width=0.4pt, fill=white,
                 minimum width=4.0cm, minimum height=1.55cm, align=center},
  tierbox/.style={draw=black, line width=0.3pt, fill=paperTone,
                  minimum width=0.72cm, minimum height=0.46cm, align=center,
                  font=\scriptsize\ttfamily},
  workerbox/.style={draw=black, line width=0.3pt, fill=paperTone,
                    minimum width=0.64cm, minimum height=0.5cm, align=center,
                    font=\scriptsize\ttfamily},
  arrow/.style={-{Stealth[length=1.8mm, width=1.5mm]}, line width=0.55pt},
  colhead/.style={font=\footnotesize\ttfamily},
]

% Column headers
\node[colhead] at ( 2.65, 8.00) {PLAYERS};
\node[colhead] at ( 9.20, 8.00) {MECHANISM};
\node[colhead] at (15.30, 8.00) {RESOURCES};

\draw[black, densely dotted, line width=0.3pt, opacity=0.45]
  (5.45, 7.75) -- (5.45, 0.60);
\draw[black, densely dotted, line width=0.3pt, opacity=0.45]
  (12.90, 7.75) -- (12.90, 0.60);

% Row 1 — Game 1 : Gamma_PD (resource allocation)
\begin{scope}
  \fill[gOne!7]  (0.3, 5.90) rectangle (17.5, 7.70);
  \draw[gOne, densely dashed, line width=0.45pt]
                 (0.3, 5.90) rectangle (17.5, 7.70);

  \node[fill=gOne, text=white, font=\scriptsize\ttfamily,
        inner xsep=7pt, inner ysep=3.5pt, anchor=north west]
        at (0.3, 7.82) {\strut GAME 1 \;$\cdot$\; $\Gamma_{\mathrm{PD}}$
                        \;\textnormal{\itshape resource allocation}};

  \node[playerbox] (P)  at (1.55, 6.80) {\textbf{Prefill $P$}};
  \node[playerbox] (D)  at (3.75, 6.80) {\textbf{Decode $D$}};

  \node[mechbox] (Planner) at (9.20, 6.80)
       {\textbf{Planner}\\[1pt]
        {\scriptsize\itshape\color{black!55}time-series forecast\,$\cdot$\,$\pm 1$ GPU}};

  \node[resbox] (GPU) at (15.30, 6.80) {\textbf{GPU budget $G$}};

  \draw[arrow, gOne] (D.east)       -- (Planner.west);
  \draw[arrow, gOne] (Planner.east) -- (GPU.west);
\end{scope}

% Row 2 — Game 2 : Gamma_KV (KV cache placement)
\begin{scope}
  \fill[gTwo!7]  (0.3, 3.30) rectangle (17.5, 5.10);
  \draw[gTwo, densely dashed, line width=0.45pt]
                 (0.3, 3.30) rectangle (17.5, 5.10);

  \node[fill=gTwo, text=white, font=\scriptsize\ttfamily,
        inner xsep=7pt, inner ysep=3.5pt, anchor=north west]
        at (0.3, 5.22) {\strut GAME 2 \;$\cdot$\; $\Gamma_{\mathrm{KV}}$
                        \;\textnormal{\itshape KV cache placement}};

  \node[playerwide] (W) at (2.65, 4.20)
       {\textbf{GPU workers $\{1,\ldots,n\}$}};

  \node[mechbox] (KVBM) at (9.20, 4.20)
       {\textbf{KV Block Manager}\\[1pt]
        {\scriptsize\itshape\color{black!55}frequency eviction\,$\cdot$\,Stackelberg leader}};

  \node[resbox] (Mem) at (15.30, 4.20) {};
  \node[anchor=center, font=\small\bfseries] at (15.30, 4.75)
       {tiered memory $G_1$--$G_4$};
  \node[tierbox] at (13.95, 3.85) {G1};
  \node[tierbox] at (14.75, 3.85) {G2};
  \node[tierbox] at (15.55, 3.85) {G3};
  \node[tierbox] at (16.35, 3.85) {G4};

  \draw[arrow, gTwo] (W.east)    -- (KVBM.west);
  \draw[arrow, gTwo] (KVBM.east) -- (Mem.west);
\end{scope}

% Row 3 — Game 3 : Gamma_R (request routing)
\begin{scope}
  \fill[gThree!7] (0.3, 0.70) rectangle (17.5, 2.50);
  \draw[gThree, densely dashed, line width=0.45pt]
                  (0.3, 0.70) rectangle (17.5, 2.50);

  \node[fill=gThree, text=white, font=\scriptsize\ttfamily,
        inner xsep=7pt, inner ysep=3.5pt, anchor=north west]
        at (0.3, 2.62) {\strut GAME 3 \;$\cdot$\; $\Gamma_{\mathrm{R}}$
                        \;\textnormal{\itshape request routing}};

  \node[playerwide] (Q) at (2.65, 1.60)
       {\textbf{requests $Q=\{q_1,\ldots,q_m\}$}};

  \node[mechbox] (Router) at (9.20, 1.60)
       {\textbf{Smart Router}\\[1pt]
        {\scriptsize\itshape\color{black!55}sequential greedy $\equiv$ best-response}};

  \node[resbox] (Workers) at (15.30, 1.60) {};
  \node[anchor=center, font=\small\bfseries] at (15.30, 2.15)
       {decode workers $W$};
  \node[workerbox] at (13.70, 1.25) {$w_1$};
  \node[workerbox] at (14.40, 1.25) {$w_2$};
  \node[workerbox] at (15.10, 1.25) {$w_3$};
  \node[workerbox] at (15.80, 1.25) {$w_4$};
  \node[workerbox] at (16.60, 1.25) {$w_m$};

  \draw[arrow, gThree] (Q.east)      -- (Router.west);
  \draw[arrow, gThree] (Router.east) -- (Workers.west);
\end{scope}

\end{tikzpicture}%
}% end \resizebox

\caption{\textbf{The three games recast as a mechanism-design table.}
Each row names the \emph{players}, the \dynamo \emph{mechanism} that
arbitrates between them, and the \emph{resources} being contested.
Reading down the middle column recovers the standard \dynamo
architecture---Planner, KVBM, Smart Router---but each component
now carries its game-theoretic role as an explicit caption.
Saturation is what binds the three rows: when one row's cost function
approaches its pole, the externalities propagate across rows.}
\label{fig:three-games}
\end{figure}
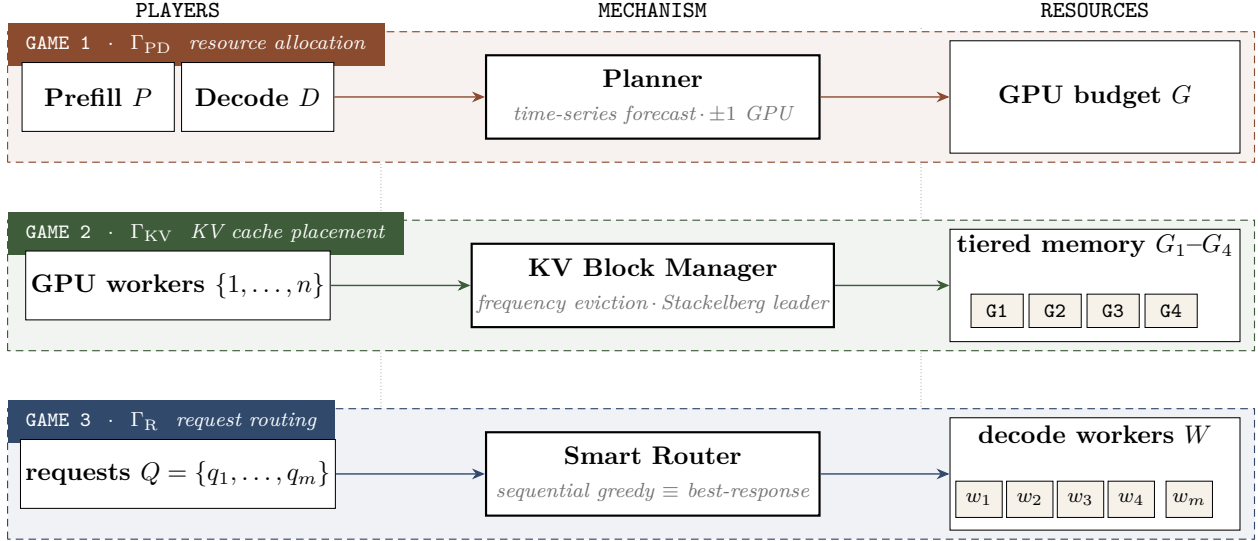

\section{Background}
\label{sec:background}

\subsection{Disaggregated Inference}
\label{sec:bg-disaggregated}

Autoregressive LLM inference has two computationally distinct phases.
\textbf{Prefill} processes the entire input prompt in a single forward pass, producing the KV cache entries for all input tokens; this phase is \emph{compute-bound}, limited by the GPU's peak FLOPS.
\textbf{Decode} generates output tokens one at a time, each reading the full KV cache but computing attention over only the new token; this phase is \emph{memory-bandwidth-bound}, limited by HBM read throughput.

Traditional co-located serving runs both phases on the same GPU, creating an inherent resource conflict: prefill's compute-intensity and decode's memory-bandwidth-intensity have different scaling characteristics, different optimal batch sizes, and different sensitivity to GPU count.
Disaggregated serving~\citep{zhong2024distserve, patel2024splitwise} resolves this by physically separating the two phases onto distinct GPU pools, connected by a high-bandwidth KV cache transfer layer.
A request arriving at the system is first routed to a prefill worker, which processes the prompt and produces the KV cache.
The KV cache is then transferred (via NVLink, InfiniBand, or RDMA) to a decode worker, which generates the output token by token.

This separation enables independent scaling: the prefill pool can be sized to meet time-to-first-token (TTFT) targets while the decode pool is sized for inter-token latency (ITL) and throughput.
NVIDIA reports a 30$\times$ throughput improvement for DeepSeek-R1 671B on a disaggregated GB200 NVL72 rack compared to traditional co-located serving~\citep{nvidia2025dynamo}.

\subsection{NVIDIA Dynamo Architecture}
\label{sec:bg-dynamo}

\dynamo is NVIDIA's production framework for disaggregated LLM inference.
Its architecture comprises four key components relevant to our game-theoretic analysis:

\paragraph{The Planner.}
A central autoscaler that dynamically adjusts the ratio of prefill to decode workers.
The Planner scrapes Prometheus metrics (TTFT histograms, ITL distributions, queue depths) every 30 seconds, runs a time-series predictor (configurable as ARIMA, Prophet, or Kalman filter), and issues scaling decisions.
It is constrained to $\pm 1$ worker change per adjustment interval, with a 3-interval grace period for newly assigned decode workers.
The Planner profiles the TTFT-vs-prefill-count and ITL-vs-decode-count tradeoff curves during a pre-deployment sweep, effectively computing the response functions of a two-player game.

\paragraph{The Smart Router.}
A KV-aware request router that assigns incoming requests to workers.
The router maintains a global radix tree (KvIndexer) tracking which KV cache blocks reside on which GPUs.
For each incoming request, it computes a per-worker cost:
\begin{equation}
\label{eq:router-cost}
  c_j = \omega \cdot b_j^{\text{prefill}} + b_j^{\text{active}}
\end{equation}
where $b_j^{\text{prefill}}$ is the number of token blocks that would need to be prefilled on worker $j$ (inversely related to cache overlap), $b_j^{\text{active}}$ is the number of active decode blocks on worker $j$ (a proxy for current load), and $\omega$ is the \texttt{kv\_overlap\_score\_weight} parameter that controls the tradeoff between cache affinity and load balancing.
The router selects the worker with the lowest cost when the \texttt{router\_temperature} $\tau = 0$ (deterministic), or samples from a softmax distribution over costs when $\tau > 0$ (stochastic):
\begin{equation}
\label{eq:router-softmax}
  P(\text{select worker } j) = \frac{\exp(-c_j / \tau)}{\sum_{k} \exp(-c_k / \tau)}
\end{equation}

\paragraph{The KV Block Manager (KVBM).}
A hierarchical cache manager that places KV cache blocks across four tiers: GPU HBM (G1, fastest, 192\,GB per B200), CPU DRAM (G2), local SSD (G3), and networked storage (G4).
Eviction follows a frequency-based policy: each block's access frequency is initialized at 1, doubled on cache hit, and decremented by 1 per time-decay step.
Blocks with frequency $\geq 2$ are eligible for promotion from lower tiers.
Cross-node KV transfer is handled by NIXL (NVIDIA Inference Xfer Library), which supports RDMA via UCX over NVLink and InfiniBand.

\paragraph{The Event Plane.}
A messaging layer (NATS JetStream or ZeroMQ) that propagates KV cache block creation/eviction events, router-to-router synchronization messages, and worker metric updates.
Workers register in etcd with lease-backed heartbeats; lease expiry triggers automatic deregistration and failure detection.

Figure~\ref{fig:dynamo-arch} illustrates these components and their interactions.

\begin{figure}[t]
\centering
\resizebox{\textwidth}{!}{%
\begin{tikzpicture}[
  x=1cm, y=1cm,
  every node/.style={font=\small, inner sep=2pt},
  mech/.style={line width=0.85pt, fill=white,
               minimum width=3.2cm, minimum height=1.25cm, align=center},
  worker/.style={draw=black, line width=0.4pt, fill=paperTone,
                 minimum width=3.0cm, minimum height=1.15cm, align=center},
  storage/.style={draw=gTwo, densely dashed, line width=0.45pt, fill=gTwo!7,
                  minimum width=14.0cm, minimum height=1.25cm, align=center},
  events/.style={draw=black, line width=0.3pt, fill=paperTone,
                 minimum width=14.0cm, minimum height=0.95cm, align=center},
  endpoint/.style={font=\footnotesize\itshape},
  arrow/.style={-{Stealth[length=1.8mm, width=1.5mm]}, line width=0.55pt},
  ctrl/.style={arrow, densely dashed, gOne},
  dataarrow/.style={arrow, black!75},
  elabel/.style={font=\scriptsize\itshape, align=center, inner sep=1.5pt, color=black!55}
]
  % Planner (Game 1 hue: burnt sienna)
  \node[mech, draw=gOne, fill=gOne!7] (planner) at (0,3.0)
    {\textbf{Planner}\\[1pt]
     {\scriptsize\itshape\color{black!55}autoscaler\,$\cdot$\,$\sim$30\,s forecast}};

  % Pipeline row: Smart Router (Game 3) -> Prefill -> Decode
  \node[mech, draw=gThree, fill=gThree!7] (router) at (-5.6,0)
    {\textbf{Smart Router}\\[1pt]
     {\scriptsize\itshape\color{black!55}KV-aware\,$\cdot$\,softmax}};
  \node[worker] (prefill) at (0.0,0)
    {\textbf{Prefill Workers}\\[1pt]
     {\scriptsize\itshape\color{black!55}compute-bound}};
  \node[worker] (decode) at (5.6,0)
    {\textbf{Decode Workers}\\[1pt]
     {\scriptsize\itshape\color{black!55}bandwidth-bound}};

  \node[endpoint, left=0.2cm of router]  (client) {Client};
  \node[endpoint, right=0.2cm of decode] (tokens) {Tokens};

  % KVBM (Game 2 hue: moss; dashed like row bands in Fig. 1)
  \node[storage, below=1.25cm of prefill] (kvbm)
    {\textbf{KV Block Manager}\\[1pt]
     {\scriptsize\itshape\color{black!55}G1 HBM $\rightarrow$ G2 DRAM $\rightarrow$ G3 SSD $\rightarrow$ G4 networked\,$\cdot$\,frequency eviction}};

  % Event plane (neutral infrastructure)
  \node[events, below=0.45cm of kvbm] (events)
    {\textbf{Event Plane}\,{\scriptsize\itshape\color{black!55}(NATS JetStream / ZeroMQ\,$\cdot$\,etcd registry)}};

  % Data plane (request flow, black)
  \draw[dataarrow] (client)  -- (router);
  \draw[dataarrow] (router)  -- (prefill)
    node[elabel, midway, above] {route};
  \draw[dataarrow] (prefill) -- (decode)
    node[elabel, midway, above] {KV via NIXL};
  \draw[dataarrow] (decode)  -- (tokens);

  % Control plane from Planner (burnt-sienna dashed)
  \draw[ctrl] (planner.south west) --
    node[elabel, pos=0.55, above left=-1pt] {scale P/D}
    (router.north);
  \draw[ctrl] (planner.south) --
    node[elabel, midway, right] {Prometheus}
    (prefill.north);
  \draw[ctrl] (planner.south east) -- (decode.north);

  % Cache edges (navy for router lookup; moss for worker read/write)
  \draw[arrow, gThree] (router.south)  --
    node[elabel, pos=0.55, left] {lookup}
    (router.south |- kvbm.north);
  \draw[arrow, gTwo]   (prefill.south) --
    node[elabel, pos=0.55, left] {write}
    (prefill.south |- kvbm.north);
  \draw[arrow, gTwo]   (decode.south)  --
    node[elabel, pos=0.55, right] {read}
    (decode.south |- kvbm.north);

  % KVBM -> Event plane (moss dashed)
  \draw[arrow, gTwo, densely dashed] (kvbm.south) --
    node[elabel, midway, right] {create/evict events}
    (events.north);
\end{tikzpicture}%
}
\caption{\dynamo architecture for disaggregated inference. Components are colored by the game they implement, matching Figure~\ref{fig:three-games}: the \textcolor{gOne}{Planner} (burnt sienna) implements Game~1 (P/D resource allocation), adjusting the worker ratio every $\sim$30\,s from Prometheus telemetry; the \textcolor{gThree}{Smart Router} (navy) implements Game~3 (request routing), assigning requests to prefill workers via KV-cache-aware costs (Eq.~\ref{eq:router-cost}); the \textcolor{gTwo}{KV Block Manager} (moss, dashed) implements Game~2 (cache placement), staging blocks across GPU HBM, CPU DRAM, local SSD, and networked storage under a frequency-eviction policy. Prefill writes KV blocks and transfers them to decode via NIXL (RDMA over NVLink or InfiniBand). The Event Plane propagates cache create/evict events and etcd-registered worker heartbeats. Black arrows denote the request data plane; dashed burnt-sienna arrows denote Planner control flow.}
\label{fig:dynamo-arch}
\end{figure}
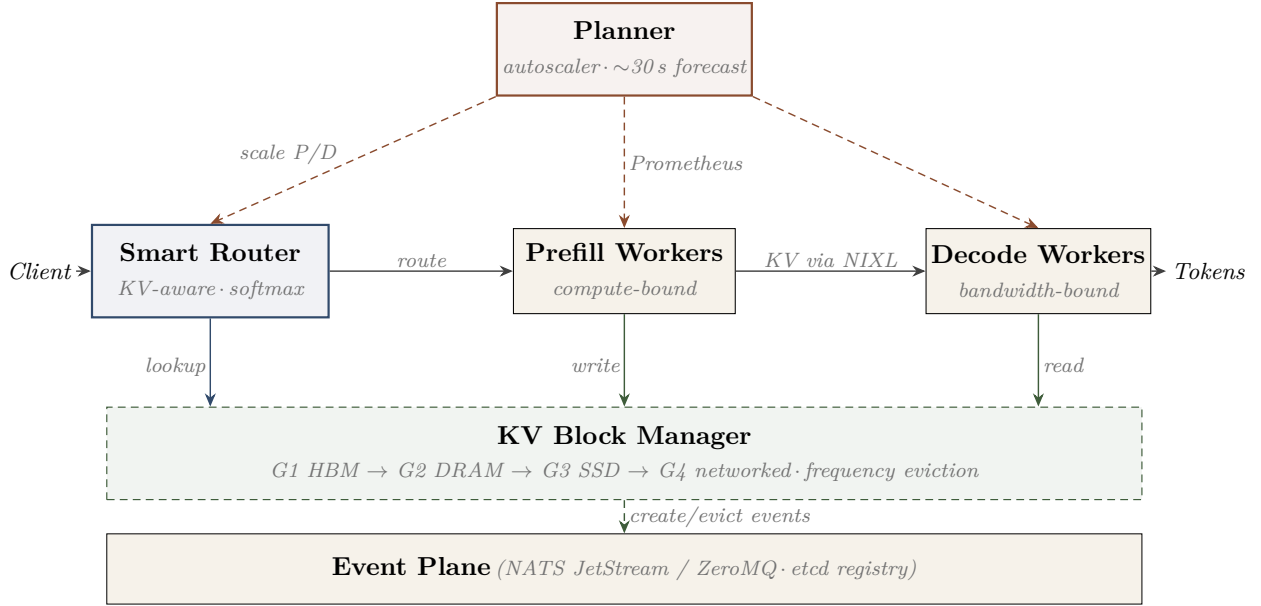

\subsection{Game Theory Preliminaries}
\label{sec:bg-gametheory}

We use standard game-theoretic concepts: normal-form games, Nash equilibrium, congestion games~\citep{rosenthal1973}, Price of Anarchy~\citep{roughgarden2002}, Pareto optimality, and Wardrop equilibrium~\citep{wardrop1952}.
Formal definitions are in Appendix~\ref{sec:appendix-defs}; we highlight only the key concepts below.

A \emph{congestion game} is one where each player's cost depends only on the \emph{count} of co-users on shared resources, not their identities, guaranteeing pure Nash equilibria via an exact potential function~\citep{monderer1996}.
The \emph{Price of Anarchy} (PoA) is the ratio of social cost at the worst Nash equilibrium to the social optimum; PoA $= 1$ means selfish behavior is optimal, while larger values indicate inefficiency.

\section{Related Work}
\label{sec:related-work}

Game theory and Pareto analysis have each been applied to GPU resource allocation, but never to inference routing---and never to each other.

\subsection{Game-Theoretic Resource Allocation in Computing Systems}
\label{sec:rw-gametheory}

Game-theoretic resource allocation is well-established at the cluster scheduling level---DRF runs in Mesos, Shockwave uses Fisher markets---but has never been applied at the inference routing level.

\textbf{Dominant Resource Fairness} (DRF)~\citep{ghodsi2011} is the most widely deployed game-theory-inspired allocator, running in Apache Mesos and YARN.
DRF satisfies four game-theoretic properties (sharing incentive, strategy-proofness, envy-freeness, and Pareto efficiency) for multi-resource allocation.
For LLM inference, DRF's multi-resource framework directly applies when prefill and decode pools compete for GPU compute, GPU memory, and network bandwidth.
DRF runs in $O(nm)$ time via progressive filling, deliberately avoiding equilibrium computation.
However, \citet{fikioris2024} proved that extending DRF to dynamic demands loses incentive compatibility, a property that holds only in static settings.
This result has direct implications for disaggregated serving, where demand patterns shift as workload composition changes.

\textbf{Themis}~\citep{mahajan2020} introduced auction-based GPU scheduling with formal fairness guarantees.
Its multi-round partial allocation auction is strategy-proof and Pareto efficient: workloads bid on GPU resources through a central arbiter, achieving a 2.25$\times$ fairness improvement over existing schedulers on Microsoft production traces.
Themis demonstrated that DRF fails for ML workloads due to gang scheduling requirements and placement sensitivity, motivating richer auction mechanisms.

\textbf{Shockwave}~\citep{zheng2023shockwave} modeled GPU scheduling as a Volatile Fisher Market, where jobs are buyers with budgets and GPUs are goods with prices.
Fisher market equilibrium with linear utilities is solvable via convex optimization (Nash social welfare maximization) in polynomial time; the authors deliberately chose a market structure that avoids PPAD-hard equilibrium computation.
Shockwave improved makespan by 1.3$\times$ and fairness by 2$\times$.

\citet{xu2025edgellm} combine a double Dutch auction with test-time reinforcement learning for mobile edge LLM inference.
This is the closest existing work combining game-theoretic mechanism design with LLM serving, though it targets edge deployments rather than datacenter-scale disaggregated serving, and does not model the prefill-decode separation or KV cache placement as separate game-theoretic problems.

\textbf{Nash bargaining} has been applied to cloud resource allocation by~\citet{perin2019}, who showed that cooperative Nash bargaining halved the average number of tasks compared to selfish best-response allocation.
\citet{facchinei2007} provide the authoritative survey on GNEPs, establishing that shared-constraint GNEPs generically admit a continuum of equilibria and that the variational equilibrium is the standard selection.
This framework directly underpins our GNEP formulation of the prefill-decode resource game.

The most directly relevant theoretical work is \citet{gaitonde2023}, who provide the first rigorous PoA analysis for strategic queuing systems, proving that no-regret learners require $2\times$ capacity compared to centralized scheduling.
Their model of selfish routing to parallel servers with queuing delay is the closest theoretical precedent for our routing game, though they do not consider cache externalities or disaggregated architectures.

None of these systems operate at inference-request timescales.
They make allocation decisions over seconds to minutes, not at the sub-millisecond cycles that LLM schedulers demand.

\subsection{Disaggregated LLM Serving Systems}
\label{sec:rw-disaggregated}

\textbf{DistServe}~\citep{zhong2024distserve} was the first system to formally advocate disaggregating prefill and decode for goodput optimization.
DistServe assigns prefill and decode to different GPUs, enabling independent parallelism strategies (larger tensor parallelism for prefill, more replicas for decode).
\textbf{Splitwise}~\citep{patel2024splitwise} independently proposed phase splitting with a focus on heterogeneous hardware: prefill on compute-optimized GPUs and decode on memory-bandwidth-optimized GPUs.
\textbf{TaiChi}~\citep{wang2025taichi} demonstrated that neither pure aggregation nor pure disaggregation is Pareto-optimal under balanced SLO requirements, proposing an adaptive hybrid that switches between modes.
This result directly motivates game-theoretic analysis: the optimal operating point depends on workload characteristics and system state, suggesting a game where the ``strategy'' is the degree of disaggregation.

\dynamo~\citep{nvidia2025dynamo} is the most complete production implementation, integrating a Planner for dynamic P/D rebalancing, a KV-aware Smart Router, hierarchical KV cache management (KVBM), and a high-performance transfer layer (NIXL) for cross-node KV movement.
Its architecture creates the richest game-theoretic structure among existing systems: the Planner plays a resource allocation game, the router plays a congestion game, and the KVBM plays a caching game.

\subsection{Pareto Analysis in LLM Serving}
\label{sec:rw-pareto}

Pareto frontier analysis dominates LLM serving evaluation but has not been connected to game-theoretic equilibrium concepts.

NVIDIA's \textbf{AIConfigurator}~\citep{nvidia2026aiconfigurator} contains a dedicated Pareto Analyzer that enumerates all valid serving configurations (spanning tensor parallelism, expert parallelism, batch sizes, and aggregated versus disaggregated modes) to identify Pareto-optimal points across throughput and latency.
For Qwen-235B on 64 H200 GPUs, AIConfigurator generates Pareto curves comparing aggregated versus disaggregated serving in approximately 30 seconds.
It achieves up to 50\% improvement for MoE models and is now integrated into \dynamo's Planner.

\textbf{Vidur}~\citep{agrawal2024vidur} uses high-fidelity simulation to explore configuration spaces, finding that small SLO changes (20ms in time-between-tokens) can shift the cost-optimal Pareto point by 1.85$\times$.
\textbf{KV Pareto}~\citep{gokhale2025kvpareto} applies Pareto analysis specifically to KV cache optimization, computing memory-accuracy frontiers across quantization levels and chunked prefill configurations, achieving 68--78\% memory savings with only 1--3\% accuracy loss.
\textbf{FlexGen}~\citep{sheng2023flexgen} used linear programming to find Pareto-optimal offloading strategies across GPU, CPU, and disk, achieving a 100$\times$ higher throughput frontier for OPT-175B.

All deployed multi-objective techniques are simple (enumeration, simulation, Pareto filtering) because configuration spaces have been bounded enough for exhaustive search.
As disaggregation, MoE routing, and heterogeneous hardware grow configuration spaces exponentially, exhaustive enumeration will stop being feasible.

\subsection{Caching Games}
\label{sec:rw-caching}

Selfish caching games are well-understood theoretically, but no one has applied them to KV cache management in LLM inference.

The foundational work on \textbf{selfish caching games}~\citep{chun2004} models server nodes that choose whether to cache data locally (at cost $\alpha$) or access remote copies (at distance-dependent cost $d(i,j)$).
Pure Nash equilibria always exist, and the Price of Anarchy varies with network topology: $\text{PoA} = 1$ on complete graphs (selfish caching is optimal) but $\text{PoA} = O(\sqrt{n})$ on line topologies.
This result has direct implications for GPU clusters: NVLink interconnects within a node are topologically close to complete graphs, suggesting near-optimal selfish KV cache placement within a node, while cross-node placement over InfiniBand (a sparser topology) may exhibit higher inefficiency.

The \textbf{capacitated selfish replication} extension applies directly to HBM-constrained GPUs: when caches have limited capacity, polynomial-time algorithms can find Nash equilibria under hierarchical network topologies.
\citet{ma2021} proved a Braess-like \textbf{cache paradox}: adding cache nodes can worsen the PoA on directed graphs.
We discuss the implications for GPU KV cache pools in Section~\ref{sec:game2}.

For congestion games with positive externalities, \citet{dekeijzer2012} proved that finding social optima is NP-hard but admits a 2-approximation.
\citet{milchtaich1996} showed that player-specific payoff functions preclude potential functions in general, a result we apply to KV cache overlap in Section~\ref{sec:game3}.

\subsection{Summary}
\label{sec:rw-summary}

Table~\ref{tab:related-work} summarizes the positioning.

\begin{table}[t]
\centering
\caption{Positioning relative to existing work. Our contribution is the first to formalize disaggregated inference as coupled games and connect game-theoretic equilibrium analysis to the specific mechanisms of a production inference system.}
\label{tab:related-work}
\small
\begin{tabular}{@{}lccccc@{}}
\toprule
\headrow \textbf{System / Work} & \textbf{Game} & \textbf{Pareto} & \textbf{Disagg.} & \textbf{KV Cache} & \textbf{Per-request} \\
\midrule
DRF~\citep{ghodsi2011}               & \checkmark &            &            &            &            \\
Themis~\citep{mahajan2020}            & \checkmark & \checkmark &            &            &            \\
Shockwave~\citep{zheng2023shockwave}  & \checkmark &            &            &            &            \\
Edge LLM~\citep{xu2025edgellm}            & \checkmark &         &            &            & \checkmark \\
AIConfigurator~\citep{nvidia2026aiconfigurator} & & \checkmark & \checkmark &            &            \\
Vidur~\citep{agrawal2024vidur}        &            & \checkmark & \checkmark &            &            \\
KV Pareto~\citep{gokhale2025kvpareto}  &            & \checkmark &            & \checkmark &            \\
DistServe~\citep{zhong2024distserve}  &            &            & \checkmark &            &            \\
Splitwise~\citep{patel2024splitwise}  &            &            & \checkmark &            &            \\
TaiChi~\citep{wang2025taichi}         &            & \checkmark & \checkmark &            &            \\
Selfish caching~\citep{chun2004}      & \checkmark &            &            & \checkmark &            \\
\midrule
\textbf{This work}                    & \checkmark & \checkmark & \checkmark & \checkmark & \checkmark \\
\bottomrule
\end{tabular}
\end{table}

\section{Formalization: Disaggregated Serving as Coupled Games}
\label{sec:formalization}

We model disaggregated inference as three coupled games that operate at different timescales and granularities.
The prefill-decode resource game operates at the Planner's adjustment timescale (tens of seconds), the request routing game operates per-request (sub-millisecond), and the KV cache placement game operates continuously as blocks are created, evicted, and promoted.
We analyze each game independently before characterizing their coupling in Section~\ref{sec:coupling}.

\paragraph{Forward references.} This section references the superlinear latency form (Equation~\ref{eq:latency-superlinear}) and the regime-transition proposition (Proposition~\ref{prop:poa-divergence}), both of which are defined in Section~\ref{sec:saturation}. Readers who prefer to read the saturation analysis first can skip to Section~\ref{sec:saturation} and return here; the formalization is self-contained modulo those two forward pointers.

\subsection{Game 1: Prefill-Decode Resource Allocation}
\label{sec:game1}

The Planner's core decision is how to partition a fixed GPU budget between prefill and decode pools.
We model this as a two-player non-cooperative game.

\begin{definition}[Prefill-Decode Resource Game]
\label{def:pd-game}
The \emph{prefill-decode resource game} $\Gamma_{\text{PD}}$, parameterized by arrival rate $\lambda$, is a two-player game with a shared constraint~\citep{rosen1965}:
\begin{itemize}[nosep]
  \item \textbf{Players:} $N = \{P, D\}$ (prefill pool, decode pool).
  \item \textbf{Strategies:} Player $P$ claims $G_P \in [0, G]$ GPUs; player $D$ claims $G_D \in [0, G]$ GPUs.
  \item \textbf{Shared constraint:} $G_P + G_D \leq G$, where $G$ is the total GPU budget.
  \item \textbf{Utility functions:}
  \begin{align}
    u_P(G_P, G_D) &= -V_{\text{TTFT}}(G_P, \lambda) \label{eq:util-prefill} \\
    u_D(G_P, G_D) &= -V_{\text{ITL}}(G_D, \lambda, G_P) \label{eq:util-decode}
  \end{align}
  where $V_{\text{TTFT}}(G_P, \lambda)$ is the TTFT SLO violation rate given $G_P$ prefill GPUs and arrival rate $\lambda$, and $V_{\text{ITL}}(G_D, \lambda, G_P)$ is the ITL SLO violation rate.
\end{itemize}
The shared constraint couples the players' feasible sets, making $\Gamma_{\text{PD}}$ a Generalized Nash Equilibrium Problem (GNEP).
\end{definition}

The decode pool's utility depends on $G_P$ through the KV transfer rate: a starved prefill pool idles decode workers regardless of decode GPU count.
This coupling is asymmetric: prefill utility depends only on its own allocation and arrival rate, while decode utility depends on both.

\begin{proposition}[Equilibrium of the P/D game]
\label{prop:pd-equilibrium}
Under the assumptions that $V_{\text{TTFT}}$ is strictly convex decreasing in $G_P$ and $V_{\text{ITL}}$ is strictly convex decreasing in $G_D$ (diminishing returns from additional GPUs), the prefill-decode resource game has a unique variational equilibrium~\citep{rosen1965, facchinei2007} at the allocation $(G_P^*, G_D^*)$ satisfying $G_P^* + G_D^* = G$ and:
\begin{equation}
\label{eq:pd-nash}
  \frac{\partial V_{\text{TTFT}}}{\partial G_P}\bigg|_{G_P^*} = \frac{\partial V_{\text{ITL}}}{\partial G_D}\bigg|_{G_D^*}
\end{equation}
That is, the marginal SLO improvement from adding one GPU to either pool is equal at equilibrium.
\end{proposition}

\begin{proof}[Proof sketch]
KKT conditions on the shared budget constraint with a common multiplier $\mu$ give $-\partial V_{\text{TTFT}}/\partial G_P = \mu = -\partial V_{\text{ITL}}/\partial G_D$; eliminating $\mu$ yields~\eqref{eq:pd-nash}.
Uniqueness follows from strict convexity via the intermediate value theorem on the one-dimensional constraint manifold $G_P + G_D = G$: each marginal is strictly monotone in $G_P$ with opposite signs, giving a unique crossing.
We adopt the variational equilibrium~\citep{rosen1965}, the standard selection for shared-constraint GNEPs~\citep{facchinei2007}.
\end{proof}

\begin{remark}[Equilibrium vs.\ social optimum]
\label{rem:pd-efficiency}
The variational equilibrium equalizes marginal violation rates but does not account for the positive externality that prefill provides to decode: more prefill GPUs produce KV cache entries faster, reducing $V_{\text{ITL}}$.
The social optimum minimizes $V_{\text{TTFT}}(G_P) + V_{\text{ITL}}(G_D, G_P)$ and satisfies $\frac{\partial V_{\text{TTFT}}}{\partial G_P} + \frac{\partial V_{\text{ITL}}}{\partial G_P} = \frac{\partial V_{\text{ITL}}}{\partial G_D}$, allocating more GPUs to prefill by an amount proportional to $|\partial V_{\text{ITL}} / \partial G_P|$.
This gap is small below saturation (when prefill throughput exceeds decode demand) but widens at saturation, contributing to the cascading dynamics analyzed in Section~\ref{sec:saturation}.
\end{remark}

\paragraph{Connection to Dynamo's Planner.}
The Planner's iterative adjustment ($\pm 1$ worker per 30\,s interval) is a best-response dynamic with inertia that converges to the variational equilibrium under stationary load.
Our experiments use fixed P/D splits, so we do not validate this convergence directly; the fixed allocations produce stable PoA regimes consistent with the equilibrium analysis.

\subsection{Game 2: KV Cache Placement}
\label{sec:game2}

The KVBM manages KV cache blocks across a four-tier hierarchy.
We model cache placement as a selfish caching game~\citep{chun2004} extended to hierarchical topologies.

\begin{definition}[KV Cache Placement Game]
\label{def:kv-game}
The \emph{KV cache placement game} $\Gamma_{\text{KV}}$ is defined as:
\begin{itemize}[nosep]
  \item \textbf{Players:} $N = \{1, \ldots, n\}$ (GPU workers, each managing a local cache).
  \item \textbf{Resources:} A set of KV cache blocks $B = \{b_1, \ldots, b_m\}$ associated with active sequences.
  \item \textbf{Strategy:} For each block $b \in B$, worker $i$ chooses a local tier $t_i(b) \in \{G1, G2, G3\}$, accesses a remote copy from another worker (case~4 in Equation~\ref{eq:kv-cost}), or does not cache the block.
  \item \textbf{Cost:} Worker $i$'s cost for accessing block $b$ is:
  \begin{equation}
  \label{eq:kv-cost}
    c_i(b) = \begin{cases}
      \alpha_{G1} & \text{if } b \text{ is in } i\text{'s HBM (fastest)} \\
      \alpha_{G2} & \text{if } b \text{ is in } i\text{'s CPU DRAM} \\
      \alpha_{G3} & \text{if } b \text{ is on } i\text{'s local SSD} \\
      d(i, j) + \alpha_{t_j(b)} & \text{if } b \text{ is cached by worker } j \neq i \text{ at tier } t_j(b) \\
      \gamma & \text{if } b \text{ is not cached by any worker (recomputation cost)}
    \end{cases}
  \end{equation}
  where $\alpha_{G1} < \alpha_{G2} < \alpha_{G3} < \gamma$ are tier access latencies and $d(i,j)$ is the network transfer cost between workers $i$ and $j$.
  \item \textbf{Capacity constraints:} Each tier on each worker has a fixed capacity: $K_i^{G1}$ blocks in HBM, $K_i^{G2}$ in DRAM, $K_i^{G3}$ on SSD. G4 (networked) is effectively unbounded but has the highest access cost.
\end{itemize}
\end{definition}

\begin{proposition}[Equilibrium structure of the KV cache game]
\label{prop:kv-equilibrium}
The KV cache placement game has the following properties:
\begin{enumerate}[nosep]
  \item Pure Nash equilibria exist (by the results of~\citet{chun2004} for selfish caching games).
  \item On complete-graph topologies (all workers equidistant, as approximated by NVLink within a node), $\text{PoA} = 1$: selfish caching is socially optimal.
  \item Under capacity constraints on HBM, the game becomes a \emph{capacitated selfish replication} game, which admits polynomial-time Nash equilibrium computation for hierarchical topologies.
\end{enumerate}
\end{proposition}

\begin{remark}[Conjectured PoA on hierarchical topologies]
\label{rem:kv-hierarchical-conjecture}
On hierarchical topologies (NVLink within node, InfiniBand across nodes), we conjecture that the PoA scales sublinearly with worker count.
The closest classical analog is \citet{chun2004}'s $O(\sqrt{n})$ bound for line-topology selfish replication under uniform local cost; our tier-heterogeneous cost model (Equation~\ref{eq:kv-cost}) relaxes the uniform-cost assumption, and a tight bound for this setting is not established.
In the homogeneous-request regime (common model prefixes), we expect the PoA to approach 1.
\end{remark}

\paragraph{Stackelberg structure.}
The KVBM's four-tier hierarchy naturally induces a Stackelberg game: the HBM tier (fastest, most constrained) acts as a ``leader'' by setting eviction policies, and lower tiers respond optimally.
The KVBM's frequency-based eviction (frequency initialized at 1, doubled on hit, decremented on decay) approximates a threshold strategy: blocks are promoted when their frequency exceeds a threshold determined by the tier's current occupancy.
This greedy local decision mirrors the equilibrium structure of selfish caching games~\citep{chun2004}, where each node caches content that minimizes its local access cost.

\paragraph{The cache paradox.}
\citet{ma2021} proved a Braess-like result for selfish caching on directed graphs: adding cache nodes can worsen the Price of Anarchy.
This carries a concrete warning for GPU inference: naively adding GPU memory to a KV cache pool could theoretically degrade system-wide performance under selfish routing, because the additional capacity changes the equilibrium placement in ways that increase total access cost.
However, when cache request patterns are homogeneous, as they often are for popular model prefixes, the paradox is bounded.

\subsection{Game 3: Request Routing as a Congestion Game}
\label{sec:game3}

The Smart Router's per-request worker selection maps directly to a congestion game, but with an important structural enrichment: KV cache overlap creates positive externalities that break the standard congestion game framework.

\begin{definition}[Request Routing Game]
\label{def:routing-game}
The \emph{request routing game} $\Gamma_{\text{R}}$ is defined as:
\begin{itemize}[nosep]
  \item \textbf{Players:} A set of concurrent requests $Q = \{q_1, \ldots, q_m\}$.
  \item \textbf{Resources:} A set of GPU workers $W = \{w_1, \ldots, w_n\}$.
  \item \textbf{Strategy:} Each request $q_i$ selects a worker $w_j \in W$.
  \item \textbf{Cost:} The cost to request $q_i$ of being assigned to worker $w_j$ under profile $\sigma$ is:
  \begin{equation}
  \label{eq:routing-cost}
    C_i(\sigma) = \underbrace{f_j(n_j(\sigma))}_{\text{congestion cost}} - \underbrace{\omega \cdot o_{ij}}_{\text{cache externality benefit}}
  \end{equation}
  where $f_j(n_j)$ is the latency function of worker $j$ as a function of the number of requests $n_j$ assigned to it, $o_{ij} \in [0, 1]$ is the KV cache overlap score between request $q_i$ and worker $w_j$'s cached blocks, and $\omega \geq 0$ is the overlap weight (in latency units).
Note the sign convention: in Equation~\ref{eq:routing-cost}, higher $\omega$ increases the \emph{benefit} of cache overlap (reducing cost), whereas \dynamo's router cost function uses $\omega$ to weight the number of blocks \emph{requiring} prefill (increasing cost for cache misses).
These are mathematically equivalent for the $\arg\min$ worker selection (both favor cache-warm workers), but the sign is inverted relative to the implementation.
\end{itemize}
\end{definition}

\paragraph{Centralized routing as best-response dynamics.}
Although \dynamo's Smart Router is centralized, its sequential greedy assignment, selecting the lowest-cost worker for each arriving request, is equivalent to best-response dynamics in the congestion game~\citep{roughgarden2015}.
The resulting allocation is a Nash equilibrium of the sequential game (or an approximate NE when $\tau > 0$), and the PoA measures how far it falls from the globally optimal batch assignment.

\paragraph{Classical reference bounds.}
For the idealized case of atomic unsplittable routing with affine latency functions and $\omega = 0$, the PoA is bounded by $5/2$~\citep{christodoulou2005}.
In the nonatomic (Wardrop) limit, the tighter bound $\text{PoA} \leq 4/3$ applies~\citep{roughgarden2002}; large atomic games converge to this limit~\citep{cominetti2023}.
These bounds serve as \emph{reference values} for our system but do not directly apply: our latency functions include singular terms near saturation (Equation~\ref{eq:latency-superlinear}, defined in Section~\ref{sec:saturation}; the key property is a pole at capacity that makes latency grow faster than any polynomial), KV cache externalities make costs player-specific (violating the anonymous-cost assumption), and the system has only 2--5 decode workers (making worst-case bounds loose for small $n$).

\paragraph{Positive externalities break the potential game structure.}
The KV overlap term $\omega \cdot o_{ij}$ makes costs depend on request \emph{identity} (prefix overlap), not just co-user count, violating the anonymous-cost assumption~\citep{milchtaich1996}.
Finding social optima is NP-hard, though 2-approximations exist~\citep{dekeijzer2012}.
When $\omega = 0$, the game is a pure congestion game with classical guarantees; when $\omega > 0$, the equilibrium shifts between cache-affinity (low TTFT, high congestion) and load-balanced (low ITL, more cache misses).

\begin{proposition}[Routing game equilibria under cache externalities]
\label{prop:routing-equilibrium}
The request routing game $\Gamma_{\text{R}}$ with cache externalities has the following properties:
\begin{enumerate}[nosep]
  \item When $\omega = 0$, $\Gamma_{\text{R}}$ is an exact potential game. Pure Nash equilibria exist, and classical bounds apply: $\text{PoA} \leq 5/2$ for affine $f_j$ (atomic)~\citep{christodoulou2005}, tightening to $\leq 4/3$ in the nonatomic limit~\citep{roughgarden2002, cominetti2023}. For non-affine costs (e.g., the singular functions in Equation~\ref{eq:latency-superlinear}), these bounds do not hold.
  \item When $\omega > 0$ and overlap scores are heterogeneous across requests, $\Gamma_{\text{R}}$ is not a potential game in general~\citep{milchtaich1996}. However, pure Nash equilibria plausibly exist when the set of distinct prefix groups is small relative to the number of workers (each prefix group forms an internal congestion game); a formal proof of existence under general prefix-group structures remains open.
  \item Each value of $(\omega, \tau)$ induces a different equilibrium, tracing a curve on the Pareto frontier in $(\text{TTFT}, \text{ITL}, \text{throughput})$ space.
\end{enumerate}
\end{proposition}

\begin{remark}[Empirical $\widehat{\text{PoA}}$ and its relationship to classical bounds]
\label{rem:poa-omega}
Our estimated index $\widehat{\text{PoA}} \approx 7.5$ at $C=64$ (70B, 2 workers), and $\approx 19$ on the 340B, lies 3--7$\times$ above the classical affine bounds: $5/2$ for unweighted atomic congestion games~\citep{christodoulou2005} and $(3+\sqrt{5})/2 \approx 2.618$ for weighted atomic congestion games~\citep{awerbuch2005}.
Since our workload uses identical requests (same input length, deterministic generation), the unweighted bound $5/2$ is the directly applicable reference; the weighted bound $(3+\sqrt{5})/2$ would apply if requests imposed materially different loads (e.g., variable sequence lengths).
The index is invariant to $\omega$ at this concurrency level (Experiment~4a confirms $\widehat{\text{PoA}} = 7.47 \pm 0.08$ across all $(\tau, \omega)$ combinations including $\omega = 0$), so the comparison to classical bounds holds regardless of whether the measurement uses $\omega = 0$ or $\omega = 1$.
The gap reflects the uncalibrated Hungarian denominator and non-affine latency at our operating points, so classical values serve as reference points rather than constraints (Section~\ref{sec:ctrl-implementation}).
The critical result is that the index is \emph{stable} below saturation and \emph{grows rapidly} at saturation (Proposition~\ref{prop:poa-divergence}): this regime transition drives the controller design in Section~\ref{sec:controller}.
\end{remark}

\paragraph{The role of router temperature.}
The temperature parameter $\tau$ in Equation~\ref{eq:router-softmax} controls the exploration-exploitation tradeoff:
\begin{itemize}[nosep]
  \item $\tau = 0$: Deterministic, greedy selection ($\arg\min$ of cost; the softmax in Equation~\ref{eq:router-softmax} is defined in the limit $\tau \to 0$). Converges quickly to a local equilibrium but is susceptible to herding (all requests pile onto the same cache-warm worker).
  \item $\tau \to \infty$: Uniform random selection. No congestion hotspots, but no cache exploitation.
  \item Intermediate $\tau$: Softmax sampling balances cache affinity and load distribution.
\end{itemize}
A recent result on game-theoretic load balancing~\citep{fardno2025loadbalancing} proved that static congestion games converge to Nash equilibrium in at most $n$ iterations of best-response dynamics, where $n$ is the number of players.
This is fast enough for practical routing when the player count (concurrent requests) is moderate, suggesting that \dynamo's greedy routing may achieve near-equilibrium allocation within a few scheduling cycles.

\subsection{Coupling Between Games}
\label{sec:coupling}

The three games are not independent.
Their coupling creates feedback loops that can amplify or dampen inefficiencies.

The three games are coupled through shared state (Figure~\ref{fig:coupled-games}):
\begin{itemize}[nosep]
  \item \textbf{Planner $\to$ Router:} GPU allocation changes the routing game's worker set and congestion landscape.
  \item \textbf{Router $\to$ Cache:} Request placement determines which KV blocks are created on which workers; cache-affinity routing concentrates blocks, increasing eviction pressure.
  \item \textbf{Cache $\to$ Router:} Cache evictions change overlap scores, redirecting future requests.
  \item \textbf{Cascading saturation:} When the prefill pool saturates, decode workers idle, the Planner shifts GPUs to prefill, changing the routing game's worker set, a feedback loop analyzed in the next section.
\end{itemize}

\begin{figure}[t]
\centering
\begin{tikzpicture}[
  every node/.style={font=\small, inner sep=2pt},
  game/.style={densely dashed, line width=0.45pt,
               minimum width=4.2cm, minimum height=2.0cm, align=center},
  arrow/.style={-{Stealth[length=1.8mm, width=1.5mm]}, line width=0.55pt},
  elabel/.style={font=\scriptsize\itshape, align=center, inner sep=1.5pt, color=black!55}
]
  % Three games (colors match Figure \ref{fig:three-games})
  \node[game, draw=gOne,   fill=gOne!7]   (g1) at (0, 0)
       {\textbf{Game 1:} $\Gamma_{\mathrm{PD}}$\\[1pt]
        Prefill--Decode\\Resource Game\\[1pt]
        {\scriptsize\itshape\color{black!55}Planner\,$\cdot$\,$\sim$30\,s}};
  \node[game, draw=gThree, fill=gThree!7] (g3) at (7, 0)
       {\textbf{Game 3:} $\Gamma_{\mathrm{R}}$\\[1pt]
        Request Routing\\Congestion Game\\[1pt]
        {\scriptsize\itshape\color{black!55}Router\,$\cdot$\,$<$1\,ms}};
  \node[game, draw=gTwo,   fill=gTwo!7]   (g2) at (3.5, -3.5)
       {\textbf{Game 2:} $\Gamma_{\mathrm{KV}}$\\[1pt]
        KV Cache\\Placement Game\\[1pt]
        {\scriptsize\itshape\color{black!55}KVBM\,$\cdot$\,continuous}};

  % Coupling arrows (colored by originating game)
  \draw[arrow, gOne] (g1.east) -- (g3.west)
    node[elabel, midway, above] {GPU allocation\\changes worker set};
  \draw[arrow, gThree] (g3.south west) -- (g2.north east)
    node[elabel, midway, right, xshift=2mm] {Routing determines\\cache creation};
  \draw[arrow, gTwo] (g2.north west) -- (g1.south)
    node[elabel, midway, left, xshift=-2mm] {Cache pressure\\triggers scaling};
  \draw[arrow, gTwo, densely dashed] (g2.north east) to[out=30, in=-100] (g3.south east);
  \node[elabel, anchor=west] at (7.2, -2.8) {Cache state changes\\overlap scores};
\end{tikzpicture}
\caption{The three coupled games in disaggregated inference, colored to match Figure~\ref{fig:three-games}. Solid arrows denote direct coupling through shared state and are colored by the game that originates the coupling; the dashed arrow is the feedback loop where cache state affects routing decisions. The games operate at different timescales: the Planner adjusts every $\sim$30 seconds, the router makes per-request decisions in under 1\,ms, and the KVBM operates continuously.}
\label{fig:coupled-games}
\end{figure}
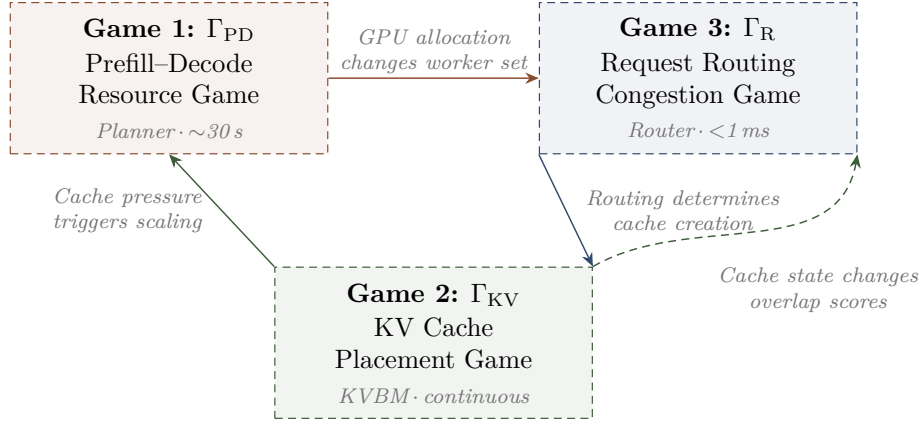

\section{Saturation Dynamics and Regime Transitions}
\label{sec:saturation}

The properties in Section~\ref{sec:formalization} assume ``well-behaved'' (linear or mildly convex) latency functions.
GPU inference exhibits sharp nonlinearities at saturation that alter the game's structure, creating regime transitions with consequences for the Price of Anarchy.
For any fixed polynomial degree $d$, the PoA is finite~\citep{aland2011}; the growth is continuous and monotone, not a discontinuity in the strict physical sense.
We use ``regime transition'' to describe the rapid, practically significant degradation that occurs as the effective degree of the latency function increases near saturation.

\subsection{Nonlinear Payoff Shifts at Saturation}
\label{sec:sat-nonlinear}

Below saturation, adding requests to a GPU has near-zero marginal latency cost: the GPU has idle compute and memory bandwidth, and batching amortizes fixed overhead.
The latency function in this regime is approximately linear:
\begin{equation}
\label{eq:latency-linear}
  f_j(n_j) \approx a_j \cdot n_j + b_j, \quad n_j \ll n_j^{\text{sat}}
\end{equation}
where $n_j^{\text{sat}}$ is worker $j$'s saturation point and $a_j, b_j$ are constants determined by model size and hardware.
In this regime, the routing congestion game is well-behaved and the PoA is bounded (Remark~\ref{rem:poa-omega}).
The key property is that the PoA is \emph{stable} below saturation; it does not grow with load.

Above saturation, the latency function becomes superlinear.
Three mechanisms drive this:

\paragraph{HBM capacity walls.}
When a worker's KV cache fills HBM, eviction to lower tiers causes a discontinuous latency jump.

\paragraph{Batch size degradation.}
Beyond the optimal batch size, increasing concurrency causes throughput to plateau while latency grows, the classic throughput-latency knee.
The latency function in this regime is well-approximated by a convex function:
\begin{equation}
\label{eq:latency-superlinear}
  f_j(n_j) \approx a_j \cdot n_j + b_j + \frac{d_j}{(n_j^{\text{sat}} - n_j)^{\beta}}, \quad n_j \to n_j^{\text{sat}}
\end{equation}
where $\beta > 0$ controls the severity of the saturation penalty and $d_j$ is a hardware-dependent constant.
The singular term $d_j / (n_j^{\text{sat}} - n_j)^\beta$ resembles queuing delay functions from traffic engineering; unlike the standard BPR (Bureau of Public Roads) polynomial, it has a pole at capacity, which drives the PoA divergence analyzed below.

\paragraph{Queuing cascades.}
When all workers are saturated, requests queue at the frontend.
Queuing latency grows as $1/(\mu - \lambda)$ as arrival rate approaches capacity.

\begin{proposition}[PoA growth at saturation]
\label{prop:poa-divergence}
For the request routing game $\Gamma_{\text{R}}$ with latency functions of the form in Equation~\ref{eq:latency-superlinear}:
\begin{enumerate}[nosep]
  \item When $n_j < n_j^{\text{sat}}$ for all workers $j$ (below saturation), the empirical PoA is stable: it does not grow with load (Remark~\ref{rem:poa-omega}).
  \item When $n_j \to n_j^{\text{sat}}$ for any worker $j$ (at saturation), the PoA grows rapidly.
  Two mechanisms contribute: (a)~\emph{routing inefficiency}, where the singular latency function amplifies even small load imbalances, and (b)~\emph{resource allocation failure}, where the P/D split (Game~1) becomes the binding constraint, and no per-request routing decision can compensate for insufficient capacity.
  In disaggregated serving, the latter dominates when the prefill pool saturates first (Section~\ref{sec:sat-coupled}).
  \item The transition from stable to rapidly growing PoA is detectable via the second derivative of aggregate latency with respect to load: $\frac{d^2 \bar{L}}{d\lambda^2} > \theta$ signals the onset of saturation, where $\theta$ is a system-dependent threshold.
\end{enumerate}
\end{proposition}

\begin{proof}[Proof sketch]
\emph{(i)} When $n_j < n_j^{\text{sat}}$ for all $j$, the singular term $d_j / (n_j^{\text{sat}} - n_j)^\beta$ is bounded and the latency function is dominated by its linear component $a_j n_j + b_j$.
The PoA is stable across load levels (Remark~\ref{rem:poa-omega}).

\emph{(ii)} The latency function in Equation~\ref{eq:latency-superlinear} has a singularity at $n_j = n_j^{\text{sat}}$, making it superpolynomial near capacity.
Direct queueing-theoretic PoA results apply to this regime: \citet{gilboafreedman2014} proved that the price of anarchy in the Markovian single-server queue diverges as $\rho \to 1$, and \citet{havivroughgarden2007} extended this to exponential multi-server stations, showing $\text{PoA} \leq m$ below capacity with unbounded growth at the pole.
For polynomial latency functions of degree $d$, a separate bound $\Theta(d^d)$ holds~\citep{aland2011}; this does not directly cover our singular cost but captures the qualitative behavior of increasing-degree polynomial approximations.
The latency function has a pole at capacity, and for any finite PoA bound $B$, there exists a load level close enough to saturation where the actual PoA exceeds $B$.
The singularity drives arbitrarily large cost ratios between near-balanced NE allocations and the social optimum that leaves headroom below the pole, matching the queueing-theoretic $\rho \to 1$ divergence.
The practical conclusion is that as load approaches capacity, the PoA grows faster than any fixed polynomial bound, which our experiments confirm (Section~\ref{sec:res-equilibrium}).
Note that on identical machines with increasing cost functions, all Nash equilibria are approximately balanced (loads differ by at most 1).
The divergence arises not from load imbalance across workers, but from the \emph{sensitivity} of the singular latency function: even a near-balanced NE incurs cost disproportionately higher than the social optimum that leaves headroom below the singularity.

\emph{(iii)} The saturation signal follows from differentiating the latency function: $f_j''(n_j) = \beta(\beta+1) d_j / (n_j^{\text{sat}} - n_j)^{\beta+2}$, which diverges as $n_j \to n_j^{\text{sat}}$.
The aggregate second derivative $d^2\bar{L}/d\lambda^2$ inherits this divergence, and a finite threshold $\theta$ on this quantity defines the transition point.
\end{proof}

Below saturation, $\widehat{\text{PoA}}$ reflects a structural constant ($\approx 7$ on 2 workers for the 70B, $\approx 19$ for the 340B; these are index values under our uncalibrated cost model, not absolute efficiency ratios).
Above saturation, it grows rapidly as (a)~superlinear latency amplifies routing suboptimalities and (b)~P/D resource allocation becomes the dominant inefficiency source, as the single prefill worker queues requests regardless of decode balance.

For \emph{polynomial} latency functions, \citet{colinibaldeschi2020} proved that PoA converges to 1 in both light and heavy traffic, with worst-case PoA at intermediate demand.
This does not contradict our analysis: the singularity at $n_j^{\text{sat}}$ in Equation~\ref{eq:latency-superlinear} is superpolynomial, and our regime transition occurs where this singular term dominates.

\subsection{Asymmetric Prefill-Decode Saturation}
\label{sec:sat-coupled}

Prefill and decode saturate on \emph{different resources at different rates}.
Prefill saturates on compute (FLOPS ceiling on the GPU's SM cores); decode saturates on memory bandwidth (HBM read throughput for KV cache attention).
On NVIDIA B200 GPUs, these ceilings are:
\begin{itemize}[nosep]
  \item Prefill: $\sim$20 PFLOPS FP4 (SM100 architecture)
  \item Decode: $\sim$8 TB/s HBM3e bandwidth
\end{itemize}

This asymmetry means the pools saturate at different request rates and with different symptoms.
More critically, saturation of one pool cascades to the other:

\begin{enumerate}[nosep]
  \item A saturated prefill pool queues incoming requests, reducing the rate at which new sequences are handed to the decode pool.
  \item Decode workers become idle, waiting for KV cache transfers that never arrive.
  \item The Planner detects underutilized decode workers and shifts GPUs to prefill.
  \item This temporarily relieves prefill saturation but reduces decode capacity.
  \item When the prefill backlog clears, the burst of new sequences overwhelms the now-smaller decode pool.
\end{enumerate}

This oscillation pattern is characteristic of coupled dynamical systems with delayed feedback.
In game-theoretic terms, the variational equilibrium of the P/D game (Proposition~\ref{prop:pd-equilibrium}) becomes \emph{unstable at saturation}: the Planner's $\pm 1$ constraint and 30-second adjustment interval create a lag that prevents tracking rapid load changes.

\subsection{KV Cache Phase Transitions}
\label{sec:sat-kv}

The KV cache placement game (Section~\ref{sec:game2}) undergoes a qualitative shift at saturation.

\paragraph{Below saturation.}
HBM has sufficient capacity for all active KV blocks.
The caching game is trivial: every block is stored in G1 (HBM), access costs are minimal, and the Price of Anarchy is 1 regardless of placement strategy.

\paragraph{At saturation.}
HBM fills, and the KVBM begins evicting blocks to G2 (DRAM) or G3 (SSD).
Each eviction decision creates a negative externality: evicting one request's KV blocks forces recomputation when that request's decode phase next needs those blocks, consuming prefill compute that could have served new requests.
The caching game transitions from a benign regime where all strategies are near-optimal to a contested regime where eviction decisions have significant social cost.

The game becomes a \textbf{congestion game with coupled resources}: memory eviction on one resource type (HBM) creates load on another (prefill compute).
This coupling is not captured by standard selfish caching models, which assume that each resource type operates independently.

\begin{proposition}[KV cache PoA regime transition]
\label{prop:kv-poa-transition}
Let $\rho = \sum_j |B_j| / \sum_j K_j^{G1}$ be the ratio of total active KV blocks to total HBM capacity across all workers.
\begin{enumerate}[nosep]
  \item When $\rho < 1$ (below capacity): $\text{PoA}_{\text{KV}} = 1$.
  \item When $\rho \geq 1$ (at capacity): $\text{PoA}_{\text{KV}}$ scales sublinearly with worker count---classically $O(\sqrt{n})$ on line topologies under uniform local cost~\citep{chun2004}, though our tier-heterogeneous model (Equation~\ref{eq:kv-cost}) relaxes that assumption---while approaching 1 on well-connected topologies (NVLink within a node).
  \item The transition at the first eviction ($\rho$ crossing 1) is sharp: the PoA jumps from 1 to a value determined by the eviction policy and network topology.
\end{enumerate}
\end{proposition}

\subsection{The Pareto Frontier Shifts at Saturation}
\label{sec:sat-pareto}

The Pareto frontier relating throughput, latency, and parameter configuration (Section~\ref{sec:introduction}) is not static; it deforms under saturation.

Below saturation, the frontier is relatively flat: small throughput gains cost small latency increases, and the system operates in a region where most parameter configurations achieve similar performance.
The AIConfigurator's enumeration approach works well here because the frontier is smooth and configurations are easy to distinguish.

At saturation the frontier steepens (small throughput gains cost large latency increases) and becomes \emph{parameter-sensitive}: configurations near-optimal below saturation may be far from the frontier above it.
A single Pareto-optimal configuration computed at deployment time (as AIConfigurator does) can therefore become suboptimal as load rises.
The empirical heatmaps in Section~\ref{sec:res-pareto-saturated} (Figure~\ref{fig:pareto-heatmap}) show this directly: a uniformly flat landscape at $C=64$ becomes rugged at $C=128$.

\begin{remark}
The ideal system would dynamically select among regime-specific Pareto frontiers, precisely the goal of the adaptive controller in Section~\ref{sec:controller}.
\end{remark}

\section{Adaptive Controller Design}
\label{sec:controller}

Since the Price of Anarchy is bounded below saturation but degrades rapidly above it (Sections~\ref{sec:formalization}--\ref{sec:saturation}), a static deployment-time configuration is insufficient: the system needs a controller that detects regime changes and adapts routing parameters.

The controller operates entirely in Python, requires no modifications to \dynamo's Rust core, and leverages \texttt{router\_config\_override} to adjust parameters per-request.

\subsection{Design Principles}
\label{sec:ctrl-principles}

The controller is guided by three principles derived from the game-theoretic analysis:

\begin{enumerate}[leftmargin=*]
  \item \textbf{Regime detection, not equilibrium computation.}
  Computing Nash equilibria at runtime is intractable (Section~\ref{sec:introduction}).
  Instead, the controller detects which \emph{regime} the system is operating in (below saturation, at saturation, or above saturation) and applies parameter settings known to be near-optimal for that regime.
  This reduces the problem from PPAD-hard equilibrium computation to a classification task over aggregate metrics.

  \item \textbf{Track a robust saturation signal.}
  The theoretically motivated signal is the second derivative of latency with respect to load (Proposition~\ref{prop:poa-divergence}), which is hardware-agnostic.
  We attempted to track $d^2 \bar{L}/d\lambda^2$ directly via finite differences but found that 5\,s polling produces estimates dominated by sampling noise.
  The deployed signal is an EWMA of TTFT P99 (Section~\ref{sec:ctrl-detection}), which smooths transient spikes while tracking sustained increases.
  This trades the hardware-agnostic theoretical invariant for a model-specific threshold ($\theta_1$ scaled to the model's baseline TTFT) but fires reliably across both models in our experiments.

  \item \textbf{Exploit regime-specific PoA bounds.}
  Below saturation, the PoA is bounded (Remark~\ref{rem:poa-omega}), meaning the efficiency loss from selfish routing is tolerable, and more than offset by the TTFT reduction from KV cache hits.
  At saturation, rapidly growing PoA means cache-affinity routing can create severe hotspots.
  The controller shifts from exploitation (cache affinity) to coordination (load balancing) as the PoA bound degrades.
\end{enumerate}

\subsection{Saturation Detection}
\label{sec:ctrl-detection}

The controller monitors a sliding window of TTFT P99 and ITL P99 values, sampled from \dynamo's Prometheus metrics every 5 seconds.
Following principle~2 above, the deployed signal is an EWMA of TTFT P99:
\begin{equation}
\label{eq:saturation-signal}
  \bar{L}(t) = \alpha \cdot L(t) + (1 - \alpha) \cdot \bar{L}(t - \Delta)
\end{equation}
where $L(t)$ is the TTFT P99 at time $t$, $\Delta$ is the 5-second polling interval, and $\alpha = 0.3$ controls responsiveness.
A sustained rise in $\bar{L}(t)$ above $\theta_1$ indicates the system has crossed the saturation knee where Proposition~\ref{prop:poa-divergence} predicts PoA degradation.

The saturation state is classified into three regimes:
\begin{equation}
\label{eq:saturation-state}
  \text{regime}(t) = \begin{cases}
    \text{\textsc{Below}} & \text{if } \bar{L}(t) < \theta_1 \\
    \text{\textsc{Transition}} & \text{if } \theta_1 \leq \bar{L}(t) < \theta_2 \\
    \text{\textsc{Saturated}} & \text{if } \bar{L}(t) \geq \theta_2
  \end{cases}
\end{equation}

where $\theta_1$ and $\theta_2$ are model-dependent thresholds determined empirically (Section~\ref{sec:experiments}).
For the 70B model, $\theta_1 = 300$\,ms and $\theta_2 = 2$\,s; for the 340B, whose baseline TTFT is $\sim$150--200\,ms (vs.\ $\sim$55\,ms for the 70B), we use $\theta_1 = 1.0$\,s and $\theta_2 = 10.0$\,s.
In practice, $\theta_1$ should be set as a multiple ($\sim$3--5$\times$) of the model's baseline TTFT P99 rather than as an absolute value.
A hysteresis margin prevents oscillation at regime boundaries: the transition from \textsc{Below} to \textsc{Transition} requires $\bar{L}(t) > \theta_1$ for $k$ consecutive samples, while the reverse requires $\bar{L}(t) < \theta_1 - \epsilon$ for $k$ samples.

\subsection{Parameter Adaptation Strategy}
\label{sec:ctrl-adaptation}

Based on the saturation regime, the controller adjusts two parameters via \texttt{router\_config\_override} (Table~\ref{tab:regime-params}):

\begin{table}[h]
\centering
\caption{Parameter settings per saturation regime. The \textsc{Below} and \textsc{Transition} rows were exercised in Experiment~3; the \textsc{Saturated} row (\S) is \emph{conjectural}---its parameters were interpolated from the $\tau \in \{0.0, 0.3, 0.7, 1.0\}$ Pareto sweep and never fired in any reported experiment, because the $C{=}32 \to 128 \to 32$ load spike did not persist long enough to cross $\theta_2$. See Section~\ref{sec:ctrl-adaptation}, ``Saturated.''}
\label{tab:regime-params}
\small\begin{tabular}{@{}lccp{6cm}@{}}
\toprule
\headrow \textbf{Regime} & \textbf{Temperature $\tau$} & \textbf{Overlap weight $\omega$} & \textbf{Rationale} \\
\midrule
\textsc{Below}      & $0.0$ & $1.0$ & Exploit cache locality ($\widehat{\text{PoA}}$ bounded) \\
\textsc{Transition} & $0.7$ & $1.0$ & Calibrated optimum from 70B 1P/5D sweep (Exp.~4b) \\
\textsc{Saturated}\textsuperscript{\S} & $0.8$ & $0.1$ & \emph{Conjectural} (not fired in any reported experiment)---prioritize load balancing ($\widehat{\text{PoA}}$ grows rapidly) \\
\bottomrule
\end{tabular}
\end{table}

The rationale for each regime maps directly to the game-theoretic analysis:

\paragraph{Below saturation ($\tau=0, \omega=1$).}
Latency functions are approximately linear, so the routing game's PoA is bounded (Remark~\ref{rem:poa-omega}).
The bounded worst-case inefficiency from selfish routing is more than offset by the TTFT reduction from KV cache hits.
Deterministic, cache-affinity routing ($\tau = 0$, $\omega = 1$) maximizes prefix reuse.

\paragraph{Transition ($\tau=0.7, \omega=1.0$).}
The system is approaching saturation.
We use the 70B 1P/5D sweep as the primary driver for parameter selection because the larger routing game ($m = 5$) offers greater parameter differentiation ($2.0\times$ spread); on 1P/5D, $(\tau=0.7, \omega=1.0)$ was among the lower-$\widehat{\text{PoA}}$ cells.
We then apply the same $(\tau, \omega)$ to the 70B 1P/2D and 340B 1P/2D deployments.
This cross-topology transfer is a deliberate simplification: the 70B 1P/2D sweep (Table~\ref{tab:pareto-c128-70b}) has its own minimum at $(\tau=0.3, \omega=0.7) = 14.6$ (chosen-cell value 20.5), and the 340B 1P/2D sweep (Table~\ref{tab:pareto-c128-340b}) has its own minimum at $(\tau=0.3, \omega=1.0) = 26.6$ (chosen-cell value 34.4).
We adopt a single $(\tau, \omega)$ setting across topologies to test whether a \emph{single} regime-gated configuration, informed by the largest routing game, generalizes; the measured improvements on 1P/2D should therefore be read as \emph{lower bounds} on what a per-topology native sweep would achieve.
On the 340B with $m = 2$ workers, all configurations perform similarly ($1.6\times$ spread), so the penalty for cross-topology transfer is small in that regime.
The high temperature introduces stochastic load balancing while retaining full cache affinity.

\paragraph{Saturated ($\tau=0.8, \omega=0.1$).}
Latency functions are superlinear and $\widehat{\text{PoA}}$ grows rapidly (Proposition~\ref{prop:poa-divergence}), so we interpolate $\tau = 0.8$ with low overlap weight to approximate power-of-two-choices load balancing and suppress cache-warm hotspots.
This row never fired in Experiment~3 (the load spike did not persist long enough to cross $\theta_2$ under EWMA hysteresis) and should be read alongside any experiment that holds $C \geq 256$ for a longer duration.

\subsection{Implementation}
\label{sec:ctrl-implementation}

The controller wraps \dynamo's Python router handler (\texttt{KvPushRouter}) with a game-theoretic adapter (Algorithm~\ref{alg:adaptive-router}):

\begin{algorithm}[t]
\caption{Game-Theoretic Adaptive Router}
\label{alg:adaptive-router}
\SetAlgoLined
\KwIn{Incoming request with token IDs $t$}
\KwOut{Worker assignment $(w_j, \text{dp\_rank}, \text{overlap})$}
\BlankLine
$\text{sat\_state} \gets \text{saturation\_detector.current\_state}()$\;
$(\tau, \omega) \gets \text{regime\_params}[\text{sat\_state}]$\;
$\text{config} \gets \text{KvRouterConfig}(\text{temperature}=\tau, \text{overlap\_weight}=\omega)$\;
$t_0 \gets \text{now}()$\;
$(w_j, \text{dp\_rank}, \text{overlap}) \gets \text{kv\_router.best\_worker}(t, \text{router\_config\_override}=\text{config})$\;
\BlankLine
\tcp{Export game-theoretic metrics}
$\text{poa\_gauge.set}(\text{poa\_tracker.current\_poa}())$\;
$\text{saturation\_gauge.set}(\text{sat\_state.value})$\;
$\text{temperature\_gauge.set}(\tau)$\;
$\text{routing\_cost\_hist.observe}(\text{now}() - t_0)$\;
\BlankLine
\Return{$(w_j, \text{dp\_rank}, \text{overlap})$}\;
\end{algorithm}

The key implementation detail is that \dynamo's Python router handler accepts a \texttt{router\_\allowbreak{}config\_\allowbreak{}override} argument in its \texttt{best\_\allowbreak{}worker()} method.
This allows the controller to adjust temperature and overlap weight \emph{on every individual request} without restarting any component.
In our experiments, we implement a zero-downtime variant: two frontends run simultaneously with different parameter configurations (default on port~8000, optimal on port~8001), and the workload generator switches target port upon detecting the regime transition.
This eliminates the throughput penalty of a frontend restart while achieving the same parameter switch.
The saturation detector polls Prometheus every 5\,s and classifies the regime via the EWMA signal, adding negligible overhead.

\paragraph{Prometheus metrics.}
The controller exports four custom metrics via \dynamo's metrics API:
\begin{itemize}[nosep]
  \item \texttt{game\_poa}: Current estimated Price of Anarchy (gauge).
  \item \texttt{game\_saturation\_state}: Current regime (0/1/2) (gauge).
  \item \texttt{game\_router\_temperature}: Active temperature value (gauge).
  \item \texttt{game\_routing\_cost}: Per-worker routing cost distribution (histogram).
\end{itemize}

These metrics enable real-time monitoring of game-theoretic properties via Grafana.

\paragraph{PoA estimation.}
The controller estimates PoA in a sliding window by comparing the observed social cost (sum of actual per-request latencies) against a hindsight-optimal assignment computed offline:
\begin{equation}
\label{eq:poa-estimate}
  \widehat{\text{PoA}}(t) = \frac{\sum_{q \in W(t)} L_q^{\text{actual}}}{\text{OPT}(W(t))}
\end{equation}
where $W(t)$ is the set of requests completed in the window ending at time $t$, $L_q^{\text{actual}}$ is request $q$'s observed latency, and $\text{OPT}(W(t))$ is the minimum total latency achievable by reassigning those requests to workers, computed via the Hungarian algorithm on a cost matrix of estimated per-worker latencies.
Structurally, this is a competitive ratio against an offline hindsight-optimal assignment; we adopt PoA vocabulary because the paper's load-bearing claims concern equilibrium-level regime structure rather than online worst-case bounds.
Since routing is many-to-one (multiple requests per worker), the cost matrix replicates each worker column up to its capacity, yielding a square matrix for one-to-one optimal assignment; the resulting $\text{OPT}$ is a lower bound on the true many-to-one optimal.
The cost matrix uses an uncalibrated parametric model: $c_{ij} = a \cdot n_j + b + d/(C_j - n_j)^\beta - w_c \cdot o_{ij}$, with $a = 0.005$, $b = 0.020$, $d = 0.010$, $\beta = 2$, $C_j = 64$ (capacity per worker), and cache weight $w_c = 0.015$.
These parameters are not fitted to observed latencies; they define a relative efficiency index whose absolute magnitude depends on the parameter choice.
Consequently, the PoA values should be interpreted as \emph{regime indicators} (rising PoA signals saturation) and \emph{comparative metrics} (adaptive vs.\ static PoA), not as absolute efficiency ratios (a PoA of 19 does not mean ``$19\times$ worse than optimal'' in system terms).
The cost matrix uses frozen latencies from the observed allocation, ignoring how redistribution would change loads.
$\text{OPT}(W(t))$ is therefore an underestimate of the true optimal social cost, making $\widehat{\text{PoA}}$ an \emph{upper bound} on the true PoA.
This conservatism is acceptable for the controller's purpose: regime detection only needs the rate of change, and a rising $\widehat{\text{PoA}}$ corroborates the saturation signal from Equation~\ref{eq:saturation-signal} regardless of absolute calibration.

\section{Experimental Setup}
\label{sec:experiments}

We validate our framework on a 3-node GPU cluster running \dynamo with disaggregated prefill/decode serving (Table~\ref{tab:hardware}).

\subsection{Hardware}
\label{sec:exp-hardware}

\begin{table}[h]
\centering
\caption{Cluster configuration.}
\label{tab:hardware}
\begin{tabular}{@{}lp{10cm}@{}}
\toprule
\headrow \textbf{Component} & \textbf{Specification} \\
\midrule
Nodes & 3$\times$ HGX B200 \\
GPUs per node & 8$\times$ B200 SXM (192\,GB HBM3e each) \\
GPUs used per node & 4 (TP=4, 70B) or 8 (TP=8, 340B) \\
Total GPUs active & 12 (70B 1P/2D), 24 (340B 1P/2D), 24 (70B 1P/5D) \\
Intra-node interconnect & NVLink5 via NVSwitch (900\,GB/s per-GPU bidirectional) \\
Inter-node interconnect & 2$\times$ ConnectX-7 InfiniBand per node (400\,Gb/s each; native RDMA via UCX/verbs, $\sim$390\,Gb/s measured between node pairs at 97.5\% line rate) \\
CPU & AMD EPYC, 248 cores per node \\
System memory & 2826\,GB DDR5 per node \\
\bottomrule
\end{tabular}
\end{table}

Node~3 carries pkey \texttt{0x8060} at index~1 (RDMA tools require \texttt{--pkey\_index 1}, child interface \texttt{ibp14s0.8060}); nodes run Ubuntu 24.04 LTS.\ifanon\else{} The cluster was provisioned as Vultr SKU \texttt{vcg-b200-248c-2826g-1536vram}.\fi

\subsection{Software Stack}
\label{sec:exp-software}

\begin{itemize}[nosep]
  \item \textbf{Inference framework:} \dynamo v0.9.0 (\url{nvcr.io/nvidia/ai-dynamo/vllm-runtime:0.9.0-cuda13})
  \item \textbf{Runtime:} vLLM backend with PagedAttention~\citep{kwon2023vllm}
  \item \textbf{KV transfer:} NIXL over UCX/verbs on InfiniBand for the data path; IPoIB (10.0.0.\{1,2,3\}/24 on \texttt{ibp14s0}) used only for the NIXL metadata side channel (\texttt{VLLM\_NIXL\_SIDE\_CHANNEL\_HOST}), not for KV bulk transfer
  \item \textbf{Coordination:} etcd (service discovery), NATS JetStream (event plane)
  \item \textbf{Monitoring:} Prometheus + NATS event correlation
  \item \textbf{Controller:} Python 3.12, \texttt{nats-py}, \texttt{scipy}
\end{itemize}

FP8 is applied at load time via vLLM's \texttt{--quantization fp8} flag, which resolves to FP8-E4M3 weights and activations with dynamic per-tensor activation scaling (no calibration set used).

\subsection{Models and Topologies}
\label{sec:exp-model}

Our primary model is \textbf{Nemotron-4-340B-Instruct} at FP8 runtime quantization ($\sim$40.9\,GiB per TP rank).
We load the community BF16 repack \texttt{mgoin/Nemotron-4-340B-Instruct-vllm} (derived from NVIDIA's Nemotron-4-340B-Instruct for vLLM compatibility) and apply FP8 at load time via vLLM's \texttt{--quantization fp8} flag; NVIDIA has not published an official FP8 Nemotron-4-340B checkpoint.
With TP=8, each worker spans all 8 GPUs on a full node, forcing every KV transfer to cross InfiniBand, representative of production-scale disaggregated deployments where workers occupy entire nodes.
We use a \textbf{1P/2D} topology: one prefill worker on node~1, two decode workers on nodes~2 and~3 ($m = 2$ resources).

To validate that the game-theoretic properties are not model-specific, we also run all experiments on \textbf{Llama-3.1-70B-Instruct-FP8}~\citep{grattafiori2024llama3} (NVIDIA's official FP8 checkpoint \texttt{nvidia/Llama-3.1-70B-Instruct-FP8}), a 70-billion parameter dense model ($\sim$70\,GB at FP8).
With TP=4, each worker spans 4 GPUs on a single node, enabling two additional topologies:

\begin{itemize}[nosep]
  \item \textbf{1P/2D} (4 GPUs per node): same minimal routing game as the 340B ($m=2$).
  \item \textbf{1P/5D} (all 8 GPUs per node): one prefill worker on node~1 (GPUs~0--3), five decode workers across all three nodes (node~1 GPUs~4--7, nodes~2--3 with two workers each, $m=5$).
  This tests whether the routing game's properties change with a larger action space.
\end{itemize}

The frontend (router) runs on node~1 alongside prefill for both models.

The P/D split is fixed within each topology; varying the split (Game~1) would require restarting backend workers across nodes.

\subsection{Workload Design}
\label{sec:exp-workload}

All experiments use a \textbf{short chat} workload: 5 prompt templates with 128 input tokens, 256 max output tokens, temperature~0.0 (deterministic generation).
Requests are generated at controlled concurrency levels using an async Python client that maintains a fixed number of in-flight requests via a semaphore.
Each concurrency level includes a 30-second ramp phase (linear increase) followed by a hold phase at the target level.
Experiment~1 includes ramp data in aggregate metrics; Experiment~2 excludes ramp data to isolate steady-state behavior.
The workload client does not set an explicit random seed; generation is deterministic (temperature~0), so run-to-run variation arises from async scheduling and request timing rather than from sampling noise.
Experiment~3's below-saturation phases characterize this variation empirically ($\sigma < 0.5\%$ on $\widehat{\text{PoA}}$, $< 5$\,ms on TTFT~P99).

\paragraph{Cool-down and outlier handling (registered \emph{a priori}).}
Multi-phase load-spike experiments (Experiment~3) include a 60\,s cool-down between iterations to drain the queue established during the saturated phase.
Per-iteration TTFT P99 values that exceed $5\times$ the within-strategy median for the same phase are flagged as cool-down artefacts, reported separately in Table~\ref{tab:adaptive-70b}'s footnote, and excluded from the strategy's reported mean only when the operational cause (incomplete drain from the prior saturated phase) is identifiable from the request timestamps.
Both with-flag and without-flag aggregates are reported in any affected case so the reader can verify the rule's effect.
This rule was registered before re-running the post-bug-fix experiments; it triggered on one iteration of the 70B 1P/5D static below-saturation phase (Section~\ref{sec:res-adaptive}, Table~\ref{tab:adaptive-70b} footnote).

\subsection{Artifact Availability}
\label{sec:exp-artifact}

\ifanon
The controller wrapper ($\sim$270 lines of Python around \dynamo's \texttt{KvPushRouter}), the measurement harness, the NATS event-correlation scripts, and the raw per-request JSON logs underlying all reported tables and figures will be released under a permissive open-source licence upon acceptance; the repository URL is withheld for double-blind review.
\else
The controller wrapper ($\sim$270 lines of Python around \dynamo's \texttt{KvPushRouter}), the measurement harness, the NATS event-correlation scripts, and the raw per-request JSON logs underlying all reported tables and figures are intended for open-source release alongside publication; this statement will be updated with the repository URL in a revised version.
\fi

\subsection{PoA Measurement}
\label{sec:exp-poa}

Measuring the Price of Anarchy requires per-request worker attribution, which \dynamo's HTTP API does not expose.
We solve this via \textbf{NATS event correlation}: \dynamo's event plane publishes \texttt{active\_sequences\_events} on NATS, containing per-request \texttt{worker\_id} assignments.
Each HTTP SSE response includes an \texttt{id: cmpl-<uuid>} field that matches the NATS event's \texttt{request\_id}.
By correlating these, we achieve 100\% per-request decode worker attribution across all experiments.

The PoA estimate uses the Hungarian algorithm on a frozen-latency cost matrix (Equation~\ref{eq:poa-estimate}); see Section~\ref{sec:ctrl-implementation} for the upper-bound interpretation and the measurement note in Section~\ref{sec:results}.

\subsection{Experiments}
\label{sec:exp-list}

We conduct four experiments:

\paragraph{Experiment 1: Equilibrium characterization.}
Sweep 14 concurrency levels (1, 2, 4, 8, 16, 32, 48, 64, 96, 128, 192, 256, 384, 512) with 120-second holds per level.
At each level we measure: TTFT and ITL distributions, PoA via NATS correlation, Nash equilibrium estimate for P/D allocation, and saturation regime classification.

\paragraph{Experiment 2: Saturation regime detection.}
Sweep 9 concurrency levels (1, 4, 8, 16, 32, 64, 128, 256, 512) with the calibrated SaturationDetector active.
Measures PoA, the detector's regime classification, the finite derivative $d(\text{TTFT P99})/d(\text{concurrency})$ at each level, and KV cache tier distribution.
Tests whether the detector correctly identifies the saturation knee observed in Experiment~1.

\paragraph{Experiment 3: Adaptive versus static routing.}
Three-phase load spike ($C = 32 \to 128 \to 32$, durations 120/180/120\,s) comparing static routing (default parameters throughout) against an adaptive controller that detects the \textsc{transition} regime and switches to the Experiment~4b optimal configuration.
We run 3 iterations of each strategy to obtain confidence intervals.
The adaptive strategy uses a \emph{dual-frontend} design: two frontends run simultaneously on ports~8000 (default: $\tau=0$, $\omega=1$) and~8001 (optimal: $\tau=0.7$, $\omega=1.0$), with the workload generator switching target port upon detection, eliminating the throughput penalty of a frontend restart.
The SaturationDetector polls Prometheus at 5\,s intervals.
Run on 1P/5D for 70B and 1P/2D for 340B.

\paragraph{Experiment 4: Pareto sweep.}
Parameter sweep over a $4 \times 4$ grid: $\tau \in \{0.0, 0.3, 0.7, 1.0\}$ and $\omega \in \{0.0, 0.3, 0.7, 1.0\}$.
For each configuration, the frontend is restarted with new CLI flags, health-checked, ramped, and held for 180~seconds.
We run the sweep at two concurrency levels:
\begin{itemize}[nosep]
  \item \textbf{Experiment 4a:} $C = 64$ (below saturation), testing parameter sensitivity under normal load.
  \item \textbf{Experiment 4b:} $C = 128$ (at saturation), testing whether parameter sensitivity emerges at saturation.
\end{itemize}

\section{Results}
\label{sec:results}

We present experiments in logical rather than numerical order: Experiments~1 and~2 establish the regime structure, Experiments~4a and~4b characterize parameter sensitivity across regimes, and Experiment~3 validates the adaptive controller whose parameter choices are informed by 4a/4b.

\paragraph{Measurement note.}
\textbf{Experiments~1, 2, 4a, and 4b are single-run measurements}; the reported values reflect a single trial at each configuration and carry unknown measurement uncertainty.
The cross-configuration spread (e.g., $\pm 0.10$ across 16 Pareto configurations in Experiment~4a) measures \emph{parameter sensitivity}, not measurement uncertainty; individual PoA values (e.g., the 18.7 plateau) could shift on rerun.
Only Experiment~3 includes repeated trials ($n=3$) with means and sample standard deviations.
With $n=3$ and $\text{df}=2$, the 95\% confidence interval multiplier is $t_{0.025,2} = 4.30$, yielding wide intervals for high-variance baselines (e.g., the 70B 1P/5D static TTFT 95\% CI spans $-5.7$ to $22.1$\,s).
The Experiment~3 below-saturation phases provide a natural estimate of run-to-run variation: $\widehat{\text{PoA}}$ standard deviations of 0.02--0.11 ($\leq$0.75\%) and TTFT P99 variation under 20\,ms (one documented 2.84\,s outlier aside), suggesting stable measurement under non-saturated conditions.
\textbf{Cross-experiment drift at saturation is non-trivial.} The same nominal 340B 1P/2D configuration at $C=128$, $\tau=0$, $\omega=1$ reports three different values across experiments: $\widehat{\text{PoA}} = 27.9$ in Experiment~1 (Table~\ref{tab:equilibrium}), $36.0$ in Experiment~4b (Table~\ref{tab:pareto-c128-340b}, top-right cell), and $33.15 \pm 1.14$ in Experiment~3's static-saturated phase (Table~\ref{tab:adaptive-340b}). The spread is $\sim$29\% across three single-run or $n{=}3$ measurements at nominally identical settings.
This spread is comparable to the $1.6\times$ cross-configuration spread used to argue 340B parameter sensitivity at saturation (Section~\ref{sec:res-pareto-saturated}), so saturation-regime parameter-sensitivity claims on the 340B should be read with this drift in mind.
All reported means are averages of per-iteration P99 values, not the P99 of pooled data.
$\widehat{\text{PoA}}$ is a regime-indicator relative efficiency index (see Section~\ref{sec:ctrl-implementation} for its definition and scope); it can fall below 1 when the Hungarian denominator imperfectly approximates the attainable optimum under current KV-cache state (documented for $m{=}5$ at $C{=}64$; see Section~\ref{sec:disc-limitations}).

\subsection{Experiment 1: Equilibrium Characterization}
\label{sec:res-equilibrium}

Table~\ref{tab:equilibrium} presents the Nemotron-4-340B results across 14 concurrency levels, with NATS-based decode worker correlation achieving 100\% match rate at all levels.

\begin{table}[t]
\centering
\caption{Equilibrium characterization: Nemotron-4-340B (FP8, TP=8, 1P/2D). TTFT and ITL are P99 values. \textsuperscript{\dag}Low-load $\widehat{\text{PoA}}$ rows ($C \leq 4$) are estimator artifacts (few in-flight requests inflate the index); they are excluded from plateau claims. \textsuperscript{\ddag}Regime labels use the 70B-calibrated threshold $\theta_1{=}300$\,ms. Under the 340B-corrected threshold $\theta_1{=}1.0$\,s (Section~\ref{sec:res-saturation}), rows $C{=}32$--$96$ would reclassify as \textsc{Below}.}
\label{tab:equilibrium}
\small
\begin{tabular}{@{}rrrrrr@{}}
\toprule
\headrow $C$ & TTFT P99 & ITL P99 & $\widehat{\text{PoA}}$ & rps & Regime \\
\midrule
1   & 74\,ms    & 19.8\,ms & 84.4\textsuperscript{\dag}  & 0.2  & Below \\
2   & 134\,ms   & 19.0\,ms & 53.6\textsuperscript{\dag}  & 0.4  & Below \\
4   & 145\,ms   & 19.2\,ms & 31.6\textsuperscript{\dag}  & 0.7  & Below \\
8   & 170\,ms   & 21.9\,ms & 19.0  & 1.3  & Below \\
16  & 181\,ms   & 21.4\,ms & 19.4  & 2.5  & Below \\
32  & 278\,ms   & 21.9\,ms & 19.5  & 4.9  & Transition\textsuperscript{\ddag} \\
48  & 366\,ms   & 22.2\,ms & 19.3  & 7.4  & Transition\textsuperscript{\ddag} \\
64  & 739\,ms   & 22.7\,ms & 18.7  & 10.2 & Transition\textsuperscript{\ddag} \\
96  & 544\,ms   & 22.5\,ms & 18.7  & 15.3 & Transition\textsuperscript{\ddag} \\
\textbf{128} & \textbf{16.2\,s} & \textbf{22.6\,ms} & \textbf{27.9} & 16.5 & \textbf{Saturated} \\
192 & 45.7\,s   & 22.4\,ms & 70.2  & 17.1 & Saturated \\
256 & 83.2\,s   & 22.3\,ms & 176.9 & 17.6 & Saturated \\
384 & 109\,s    & 23.3\,ms & 283.6 & 18.0 & Saturated \\
512 & 113\,s    & 22.2\,ms & 248.2 & 18.0 & Saturated \\
\bottomrule
\end{tabular}
\end{table}

\paragraph{Asymmetric P/D saturation.}
The dominant finding is extreme P/D asymmetry (Figure~\ref{fig:ttft-itl}), confirming Section~\ref{sec:sat-coupled}: TTFT P99 grows from 74\,ms to 113\,s ($1{,}500\times$) while ITL P99 stays at $21.7 \pm 1.3$\,ms across all 14 concurrency levels, consistent with prefill being the dominant bottleneck and decode behaving as memory-bandwidth-bound on this workload.
The 70B model corroborates this with ITL P99 flat at ${\sim}10\text{--}14$\,ms and identical TTFT explosion at saturation.

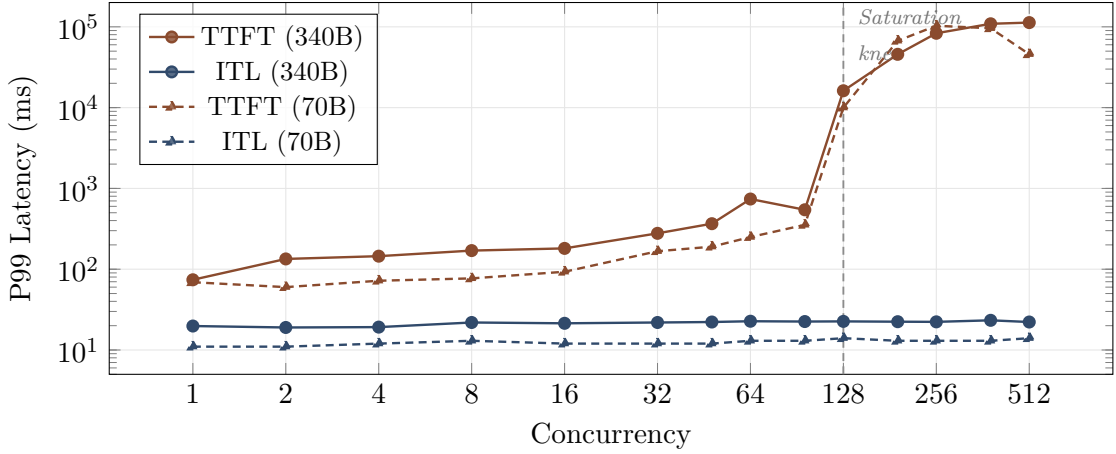
\begin{figure}[t]
\centering
\begin{tikzpicture}
\begin{axis}[
  width=0.9\columnwidth,
  height=6.5cm,
  xlabel={Concurrency},
  ylabel={P99 Latency (ms)},
  xmode=log,
  log basis x=2,
  ymode=log,
  xtick={1,2,4,8,16,32,64,128,256,512},
  xticklabels={1,2,4,8,16,32,64,128,256,512},
  ymin=5, ymax=200000,
  legend style={at={(0.03,0.97)}, anchor=north west, font=\small,
    fill=white, fill opacity=0.8, draw opacity=1, text opacity=1},
  grid=major,
  grid style={gray!20},
]
  % TTFT P99 — 340B (TP=8, 1P/2D) — primary
  \addplot[line width=0.9pt, gOne, mark=*, mark size=2pt] coordinates {
    (1,74) (2,134) (4,145) (8,170) (16,181) (32,278)
    (48,366) (64,739) (96,544) (128,16178) (192,45718)
    (256,83248) (384,109104) (512,112855)
  };
  \addlegendentry{TTFT (340B)}

  % ITL P99 — 340B (TP=8, 1P/2D) — primary
  \addplot[line width=0.9pt, gThree, mark=*, mark size=2pt] coordinates {
    (1,19.8) (2,19.0) (4,19.2) (8,21.9) (16,21.4) (32,21.9)
    (48,22.2) (64,22.7) (96,22.5) (128,22.6) (192,22.4)
    (256,22.3) (384,23.3) (512,22.2)
  };
  \addlegendentry{ITL (340B)}

  % TTFT P99 — 70B 1P/2D (corroboration)
  \addplot[line width=0.9pt, gOne, mark=triangle*, mark size=2pt, densely dashed] coordinates {
    (1,69) (2,60) (4,72) (8,77) (16,93) (32,167) (48,190)
    (64,250) (96,354) (128,10000) (192,67200) (256,103148)
    (384,96000) (512,45800)
  };
  \addlegendentry{TTFT (70B)}

  % ITL P99 — 70B 1P/2D (corroboration)
  \addplot[line width=0.9pt, gThree, mark=triangle*, mark size=2pt, densely dashed] coordinates {
    (1,11) (2,11) (4,12) (8,13) (16,12) (32,12) (48,12)
    (64,13) (96,13) (128,14) (192,13) (256,13) (384,13) (512,14)
  };
  \addlegendentry{ITL (70B)}

  % Saturation knee (first post-knee grid point; true knee lies in (96, 128])
  \draw[densely dashed, black!45, line width=0.7pt] (axis cs:128,5) -- (axis cs:128,200000);
  \node[font=\scriptsize\itshape, color=black!55, anchor=south west] at (axis cs:132,80000) {Saturation};
  \node[font=\scriptsize\itshape, color=black!55, anchor=north west] at (axis cs:132,80000) {knee};
\end{axis}
\end{tikzpicture}
\caption{Asymmetric P/D saturation: TTFT P99 versus ITL P99 for the 340B (solid) and 70B (dashed) models. Below saturation, both metrics stay at tens of milliseconds. At saturation, TTFT explodes to $\sim$100\,s while ITL remains flat, nearly four orders of magnitude apart. The 340B has a higher ITL baseline ($\sim$22\,ms vs.\ $\sim$13\,ms) reflecting its larger decode memory footprint, and a lower throughput ceiling ($\sim$18 vs.\ $\sim$47\,rps). Both models exhibit the same asymmetry, consistent with prefill as the dominant bottleneck on this workload across both model scales. The ``saturation knee'' marker indicates the first post-knee grid point ($C=128$); the true knee lies in $(96, 128]$ and is not resolvable at our grid spacing.}
\label{fig:ttft-itl}
\end{figure}

\paragraph{Throughput ceiling.}
Throughput saturates at $\sim$18\,rps for the 340B (16.5\,rps at $C=128$ vs.\ 18.0\,rps at $C=512$) and $\sim$47\,rps for the 70B, the prefill-limited ceiling.
The $2.6\times$ lower ceiling for the $4.9\times$ larger model reflects the partial offset of FP8 quantization against the increased compute cost.

\paragraph{Three PoA regimes.}
The $\widehat{\text{PoA}}$ measurements reveal three distinct regimes (Figure~\ref{fig:poa-transition}):

\begin{enumerate}[nosep]
  \item \textbf{Low load ($C = 1\text{--}4$):} $\widehat{\text{PoA}} \approx$ 31--84 (340B) or 6--36 (70B). These elevated values are artifacts: with few in-flight requests, the Hungarian algorithm finds a trivially optimal assignment that diverges from the actual routing by the full single-request cost, inflating the index because the denominator is artificially small. We de-emphasize these rows in aggregate claims (see dagger-marked entries in Table~\ref{tab:equilibrium} footnote).
  \item \textbf{Below saturation ($C = 8\text{--}96$):} $\widehat{\text{PoA}} \approx$ 18--19 (340B), 6--7.7 (70B 1P/2D), or $\sim$14.7 (70B 1P/5D), a stable plateau. The index reflects the structural cost of greedy routing and does not grow with load. This stability is consistent with Proposition~\ref{prop:poa-divergence}(i). The $2.5\times$ higher plateau for the 340B (18.7 vs.\ 7.47) reflects the larger cost of routing suboptimality with a bigger model, since each misrouted request wastes more compute.
  \item \textbf{At saturation ($C \geq 128$):} $\widehat{\text{PoA}}$ grows rapidly. The 340B peaks at $C=384$ ($\widehat{\text{PoA}} = 284$). On the 70B, the 1P/2D topology peaks at $C=384$ ($\widehat{\text{PoA}} = 199$) while the 1P/5D topology peaks earlier at $C=256$ ($\widehat{\text{PoA}} = 309$). The index decreases at $C=512$ as extreme queuing paradoxically equalizes load.
\end{enumerate}

The transition from stable to rapidly growing PoA occurs at $C=128$ for both models: on the 340B, TTFT P99 jumps from 544\,ms to 16.2\,s (${\sim}30\times$); on the 70B, from 354\,ms to 10.0\,s.
(The non-monotonic 340B TTFT P99 at $C=96$, 544\,ms, below the $C=64$ value of 739\,ms, is a single-run measurement artifact; the PoA is stable at this level.)
As discussed in Section~\ref{sec:saturation}, the deep-saturation $\widehat{\text{PoA}}$ should be decomposed: routing inefficiency (bounded, on the decode side) and resource allocation failure (dominant, the prefill bottleneck).
A $\widehat{\text{PoA}}$ of 284 at $C=384$ does not mean routing is $284\times$ suboptimal; it means the system is in overload and no routing strategy can compensate for insufficient prefill capacity.

\paragraph{Cross-model and topology consistency.}
The three-regime structure is consistent across both models and topologies (Figure~\ref{fig:poa-transition}).
On the 70B 1P/5D topology ($m = 5$), the $\widehat{\text{PoA}}$ curve has the same qualitative shape as 1P/2D---stable plateau, same post-knee grid point, and explosive growth---at a higher level throughout: the plateau sits at $\sim$14.7 ($\approx 2.4\times$ the 1P/2D 6.2--7.7), reflecting the larger 5-worker routing game, and the throughput ceiling rises to $\sim$51\,rps. At saturation the gap widens: $\widehat{\text{PoA}}$ is $3.3\times$ higher at $C=128$ (57.3 vs.\ 17.6) and $1.9\times$ higher at $C=256$ (309 vs.\ 167, both on the saturation-sweep grid), consistent with the larger routing game amplifying saturation-regime inefficiency.
The 340B at 1P/2D (TP=8, cross-InfiniBand KV transfers) reproduces the identical regime structure at a $2.5\times$ higher plateau, confirming that the game-theoretic regime structure is consistent across the tested model scales.

Table~\ref{tab:cross-model} summarizes the key cross-model comparison.

\begin{table}[t]
\centering
\caption{Cross-model properties not visible in Figures~\ref{fig:ttft-itl} and~\ref{fig:poa-transition}. \textsuperscript{*}Finite difference taken over the $[64, 128]$ interval; grid spacing cannot resolve the true knee location within $(96, 128]$.}
\label{tab:cross-model}
\small
\begin{tabular}{@{}lrrr@{}}
\toprule
\headrow \textbf{Property} & \textbf{340B (TP=8)} & \textbf{70B (TP=4)} & \textbf{Ratio} \\
\midrule
First post-knee grid point & $C=128$ & $C=128$ & Same \\
Throughput ceiling & 18\,rps & 47\,rps & 0.39$\times$ \\
$\Delta\text{TTFT}/\Delta C$ across knee\textsuperscript{*} & $\sim$0.55 & $\sim$0.46 & $\sim$Same \\
\bottomrule
\end{tabular}
\end{table}

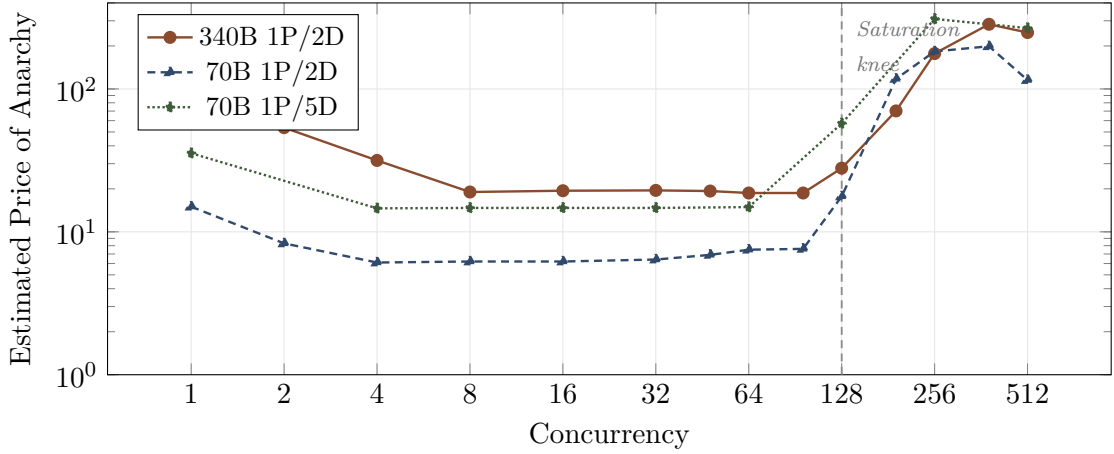
\begin{figure}[t]
\centering
\begin{tikzpicture}
\begin{axis}[
  width=0.9\columnwidth,
  height=6.5cm,
  xlabel={Concurrency},
  ylabel={Estimated Price of Anarchy},
  xmode=log,
  log basis x=2,
  xtick={1,2,4,8,16,32,64,128,256,512},
  xticklabels={1,2,4,8,16,32,64,128,256,512},
  ymode=log,
  ymin=1, ymax=400,
  legend style={at={(0.03,0.97)}, anchor=north west, font=\small},
  grid=major,
  grid style={gray!20},
]
  % Measured PoA — 340B 1P/2D (primary)
  \addplot[line width=0.9pt, gOne, mark=*, mark size=2pt] coordinates {
    (1, 84.4) (2, 53.6) (4, 31.6) (8, 19.0) (16, 19.4) (32, 19.5)
    (48, 19.3) (64, 18.7) (96, 18.7) (128, 27.9) (192, 70.2)
    (256, 176.9) (384, 283.6) (512, 248.2)
  };
  \addlegendentry{340B 1P/2D}

  % Measured PoA — 70B 1P/2D (corroboration)
  \addplot[line width=0.9pt, gThree, mark=triangle*, mark size=2pt, densely dashed] coordinates {
    (1, 15.0) (2, 8.3) (4, 6.1) (8, 6.2) (16, 6.2) (32, 6.4)
    (48, 6.9) (64, 7.5) (96, 7.6) (128, 17.7) (192, 117.6)
    (256, 183.0) (384, 199.0) (512, 115.4)
  };
  \addlegendentry{70B 1P/2D}

  % Measured PoA — 70B 1P/5D (corroboration; Experiment 2 grid)
  \addplot[line width=0.9pt, gTwo, mark=diamond*, mark size=2pt, densely dotted] coordinates {
    (1, 35.5) (4, 14.6) (8, 14.7) (16, 14.7) (32, 14.7)
    (64, 14.9) (128, 57.3) (256, 309.0) (512, 266.5)
  };
  \addlegendentry{70B 1P/5D}

  % Saturation knee annotation (first post-knee grid point; true knee lies in (96, 128])
  \draw[densely dashed, black!45, line width=0.7pt] (axis cs:128,1) -- (axis cs:128,400);
  \node[font=\scriptsize\itshape, color=black!55, anchor=south west] at (axis cs:132,200) {Saturation};
  \node[font=\scriptsize\itshape, color=black!55, anchor=north west] at (axis cs:132,200) {knee};
\end{axis}
\end{tikzpicture}
\caption{Estimated Price of Anarchy versus concurrency (log-log scale) for the 340B 1P/2D (solid), 70B 1P/2D (dashed), and 70B 1P/5D (dotted) configurations. Both models exhibit the same three-regime structure: a stable plateau below saturation, a sharp jump at the first post-knee grid point ($C = 128$), and rapid growth to $\sim$200--309 deep in saturation ($C = 256$--$384$). The 340B plateau (18.7) is $\sim$2.5$\times$ higher than the 70B 1P/2D (7.47), reflecting the larger cost of routing suboptimality with a bigger model; the 70B 1P/5D plateau ($\sim$14.7) sits $\approx 2.4\times$ above 1P/2D, reflecting the larger 5-worker routing game. The 1P/5D series is measured on the coarser saturation-sweep grid $\{1, 4, 8, \ldots, 512\}$. The true knee location is not resolvable within $(96, 128]$ at our grid spacing.}
\label{fig:poa-transition}
\end{figure}

\subsection{Experiment 2: Saturation Regime Detection}
\label{sec:res-saturation}

On the 340B model, the finite difference $\Delta = \Delta(\text{TTFT P99})/\Delta C$ taken across the knee jumps from 0.012 on the $[32, 64]$ interval to 0.554 on the $[64, 128]$ interval, a ${\sim}45\times$ increase.
On the 70B (1P/2D), the same transition produces a $148\times$ jump (from 0.003 to 0.458).
Because the concurrency grid is $\{1, 4, 8, 16, 32, 64, 128, 256, 512\}$, this quantity is a difference quotient that straddles the knee rather than a localized derivative \emph{at} the knee; a true knee anywhere in $(64, 128]$ would produce similar values.
Within that caveat, the magnitude ($\sim$0.45--0.56) is consistent across both models and topologies, and the saturation signal is unambiguous in all cases.
A denser sweep in $(96, 128]$ would be needed to localize the knee precisely; we report the coincidence as ``same first post-knee grid point'' rather than ``same knee.''

The SaturationDetector's $\theta_1 = 300$\,ms threshold, calibrated on the 70B, triggers prematurely on the 340B: at $C = 32$, the EWMA of the streamed frontend TTFT P99 signal reached ${\sim}0.37$\,s during the hold (the table reports the client-side aggregate P99, 278\,ms), crossing $\theta_1$ and producing a \textsc{Transition} classification that is too early.
This occurs because the 340B's baseline TTFT ($\sim$150--200\,ms) is higher than the 70B's ($\sim$55\,ms).
For the 340B, we recommend $\theta_1 = 1.0$\,s and $\theta_2 = 10.0$\,s; these thresholds should be set as a fraction of the model's baseline TTFT rather than as absolute values.

The below-saturation PoA is $\sim$19 on the 340B (vs.\ $\sim$6--7 on the 70B), consistent with the cross-model plateau gap observed in Section~\ref{sec:res-equilibrium}: each misrouted request on the larger model wastes more compute.
On the 70B, the detector correctly classified all below-saturation levels as \textsc{Below} and detected the transition at $C = 128$, with the EWMA reaching \textsc{Saturated} within 15\,s (3 samples at the 5\,s polling interval).

\subsection{Experiment 4a: Pareto Sweep Below Saturation}
\label{sec:res-pareto-below}

Table~\ref{tab:pareto-c64} presents the full $4 \times 4$ parameter sweep at $C = 64$ for the 340B model (shown alongside the $C = 128$ sweeps in Table~\ref{tab:pareto-sweeps}).

\paragraph{PoA invariance.}
The central finding is that PoA is \emph{statistically indistinguishable} across router parameters below saturation: $\widehat{\text{PoA}} = 18.7 \pm 0.10$ across all 16 $(\tau, \omega)$ configurations on the 340B (cross-configuration spread; single trial per cell).
Neither $\tau$ nor $\omega$ produces a measurable effect on routing efficiency.
The 70B corroborates this: $\widehat{\text{PoA}} = 7.47 \pm 0.08$ across 16 configurations on 1P/2D, and $14.93 \pm 0.06$ on 1P/5D, the higher plateau reflecting the larger 5-worker routing game.
The pattern holds across model scale, TP configuration, and number of decode workers.

\paragraph{Implications.}
At below-saturation load, router-parameter tuning yields no measurable PoA gain---tuning effort is best spent at the saturation knee or above.
The structural inefficiency relative to the hindsight-optimal assignment ($\sim$19$\times$ for the 340B, $\sim$7$\times$ for the 70B 1P/2D, $\sim$15$\times$ for the 70B 1P/5D) is a cost of greedy routing that no parameter configuration can reduce.
Reducing the PoA below saturation requires changing the assignment algorithm from greedy to coordinated batch assignment.

\paragraph{Latency and throughput.}
While PoA is invariant, TTFT P99 shows mild variation without systematic dependence on $\tau$ or $\omega$.
ITL P99 is stable at 19--23\,ms (340B) and 11--13\,ms (70B) across all configurations.
Throughput is uniformly at the prefill-limited ceiling for each model.

\subsection{Experiment 4b: Pareto Sweep at Saturation}
\label{sec:res-pareto-saturated}

Panels~(b) and~(c) of Table~\ref{tab:pareto-sweeps} present the $4 \times 4$ parameter sweeps at $C = 128$ (saturation knee) for both models.

\begin{table}[t]
\centering
\setlength{\tabcolsep}{3.5pt}
\caption{$4 \times 4$ Pareto sweeps of $\widehat{\text{PoA}}$ across $(\tau, \omega)$. (a)~340B at $C = 64$ (below saturation, $\sim$12\,rps): all 16 cells lie within $\pm 0.2$ of 18.7---no measurable parameter effect. (b)~340B at $C = 128$ (saturation, $\sim$18.9\,rps): $1.6\times$ spread, noisy but unstructured. (c)~Llama-3.1-70B 1P/2D at $C = 128$ (saturation, $\sim$47\,rps): $1.9\times$ spread with clearer structure. All configurations achieve 100\% NATS match rate. Bold values mark per-panel minima in saturated regimes; below-saturation panel~(a) has no meaningful minimum.}
\label{tab:pareto-sweeps}
\begin{subtable}[t]{0.30\textwidth}
\centering
\caption{340B, $C = 64$ (below sat.)}
\label{tab:pareto-c64}
\begin{tabular}{@{}r|cccc@{}}
\toprule
\headrow $\tau \backslash \omega$ & 0.0 & 0.3 & 0.7 & 1.0 \\
\midrule
0.0 & 18.7 & 18.6 & 18.7 & 18.7 \\
0.3 & 18.6 & 18.6 & 18.5 & 18.7 \\
0.7 & 18.7 & 18.7 & 18.9 & 18.7 \\
1.0 & 18.6 & 18.8 & 18.6 & 18.6 \\
\bottomrule
\end{tabular}
\end{subtable}\hfill
\begin{subtable}[t]{0.30\textwidth}
\centering
\caption{340B, $C = 128$ (sat.)}
\label{tab:pareto-c128-340b}
\begin{tabular}{@{}r|cccc@{}}
\toprule
\headrow $\tau \backslash \omega$ & 0.0 & 0.3 & 0.7 & 1.0 \\
\midrule
0.0 & 31.9 & 34.6 & 32.4 & 36.0 \\
0.3 & 36.5 & 42.5 & 35.7 & \textbf{26.6} \\
0.7 & 32.0 & 29.4 & 34.2 & 34.4 \\
1.0 & 37.5 & 32.7 & 34.0 & 38.5 \\
\bottomrule
\end{tabular}
\end{subtable}\hfill
\begin{subtable}[t]{0.30\textwidth}
\centering
\caption{70B 1P/2D, $C = 128$ (sat.)}
\label{tab:pareto-c128-70b}
\begin{tabular}{@{}r|cccc@{}}
\toprule
\headrow $\tau \backslash \omega$ & 0.0 & 0.3 & 0.7 & 1.0 \\
\midrule
0.0 & 19.6 & 18.4 & 25.4 & 19.1 \\
0.3 & 27.5 & 20.6 & \textbf{14.6} & 15.6 \\
0.7 & 28.2 & 26.4 & 24.6 & 20.5 \\
1.0 & 21.2 & 25.7 & 27.2 & 15.4 \\
\bottomrule
\end{tabular}
\end{subtable}
\end{table}

\paragraph{Parameter sensitivity depends on routing game size.}
The two models reveal a nuanced interaction between parameter sensitivity and the number of decode workers.
On the 70B with 1P/2D ($m = 2$ workers), clear parameter sensitivity emerges at saturation: PoA ranges from 14.6 to 28.2 ($1.9\times$ spread, mean $21.9 \pm 4.6$).
On the 340B with 1P/2D ($m = 2$, TP=8), parameter sensitivity is moderate: $1.6\times$ spread (mean $34.3 \pm 3.7$), with variation that is noisy but unstructured.
On the 70B 1P/5D ($m = 5$), the spread is $2.0\times$ (mean $57.8 \pm 10.6$), confirming that more decode workers amplify the routing game's parameter sensitivity at saturation.

\paragraph{Variance increase across configurations.}
Despite different sensitivity magnitudes, both models show the same qualitative transition: cross-configuration variance increases sharply from below saturation to at saturation.
On the 340B, $\sigma$ grows ${\sim}37\times$ (from $\pm 0.10$ at $C = 64$ to $\pm 3.7$ at $C = 128$).
On the 70B 1P/2D, $\sigma$ grows ${\sim}58\times$ (from $\pm 0.08$ to $\pm 4.6$).
This variance jump across configurations is a candidate saturation signal; confirmation as a robust trigger requires repeated trials per cell.

\paragraph{Throughput invariance.}
On both models, throughput is constant across all 16 configurations at the prefill-limited ceiling ($\sim$18.9\,rps for 340B, $\sim$47\,rps for 70B).
Routing parameters affect \emph{how} requests are distributed across decode workers (and hence tail latency), but cannot increase aggregate throughput beyond the prefill bottleneck.
On the 70B 1P/2D, TTFT P99 varies from 6.7\,s to 25.6\,s across configurations; on the 340B, from 24.7\,s to 45.6\,s.

Figure~\ref{fig:pareto-heatmap} visualizes the contrast: below saturation, the parameter landscape is uniformly flat; at saturation, it becomes rugged (70B) or noisy (340B).

\paragraph{Empirical validation of the cache placement game (Game~2).}
The overlap weight $\omega$ controls Game~2's influence on routing: $\omega = 0$ disables cache affinity (pure congestion game), while $\omega = 1$ maximizes it.
Slicing the Pareto sweep along the $\omega$ dimension isolates the cache game's effect.
\emph{Below saturation} ($C = 64$), PoA is invariant to $\omega$ on both models (CV $< 0.5\%$), confirming Proposition~\ref{prop:kv-poa-transition}: with spare capacity, selfish cache-aware routing is as efficient as cache-oblivious routing, and the cache game's PoA contribution is negligible.
\emph{At saturation} ($C = 128$), the cache game creates a measurable tradeoff on the 70B 1P/2D: averaging across $\tau$, TTFT P99 drops $1.4\times$ as $\omega$ increases from 0 to 1 (23.7\,s $\to$ 17.4\,s), because cache-affine routing reduces redundant prefill recomputation.
However, PoA does not decrease monotonically with $\omega$, because concentrating requests on cache-warm workers creates load imbalance; the congestion externality offsets the cache benefit, exactly the tradeoff predicted by the coupling between Games~2 and~3 (Section~\ref{sec:coupling}).
On the 340B ($m = 2$), TTFT P99 is invariant to $\omega$ at saturation ($\sim$30\,s across all values), consistent with the constrained routing game having too few workers for cache affinity to differentiate outcomes.

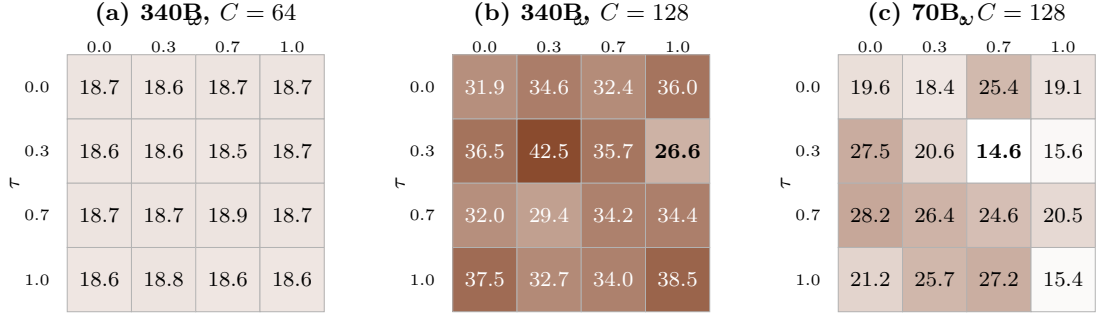
\begin{figure}[t]
\centering
\begin{tikzpicture}[every node/.style={font=\small}]
\def\s{0.85}% cell size in cm (shrunk so three panels fit side-by-side)
\def\hc#1#2#3#4#5{% col, row, value, sienna-intensity, text-color
  \fill[gOne!#4!white] ({#1*\s},{#2*\s}) rectangle ++(\s,\s);%
  \draw[black!30, line width=0.3pt] ({#1*\s},{#2*\s}) rectangle ++(\s,\s);%
  \node[text=#5, font=\scriptsize] at ({#1*\s+\s/2},{#2*\s+\s/2}) {#3};%
}
% --- Panel (a): 340B, C = 64 ---
\begin{scope}
  \node[font=\footnotesize\bfseries, anchor=south] at ({2*\s},{4*\s+0.25}) {(a) 340B, $C = 64$};
  \foreach \x/\lbl in {0/0.0, 1/0.3, 2/0.7, 3/1.0}
    \node[font=\tiny] at ({\x*\s+\s/2},{4*\s+0.1}) {\lbl};
  \node[font=\scriptsize] at ({2*\s},{4*\s+0.45}) {$\omega$};
  \foreach \y/\lbl in {3/0.0, 2/0.3, 1/0.7, 0/1.0}
    \node[font=\tiny, anchor=east] at (-0.1,{\y*\s+\s/2}) {\lbl};
  \node[font=\scriptsize, rotate=90, anchor=south] at (-0.5,{2*\s}) {$\tau$};
  % tau=0.0
  \hc{0}{3}{18.7}{15}{black}\hc{1}{3}{18.6}{14}{black}\hc{2}{3}{18.7}{15}{black}\hc{3}{3}{18.7}{15}{black}
  % tau=0.3
  \hc{0}{2}{18.6}{14}{black}\hc{1}{2}{18.6}{14}{black}\hc{2}{2}{18.5}{14}{black}\hc{3}{2}{18.7}{15}{black}
  % tau=0.7
  \hc{0}{1}{18.7}{15}{black}\hc{1}{1}{18.7}{15}{black}\hc{2}{1}{18.9}{15}{black}\hc{3}{1}{18.7}{15}{black}
  % tau=1.0
  \hc{0}{0}{18.6}{14}{black}\hc{1}{0}{18.8}{15}{black}\hc{2}{0}{18.6}{14}{black}\hc{3}{0}{18.6}{14}{black}
\end{scope}
% --- Panel (b): 340B, C = 128 ---
\begin{scope}[xshift=5.1cm]
  \node[font=\footnotesize\bfseries, anchor=south] at ({2*\s},{4*\s+0.25}) {(b) 340B, $C = 128$};
  \foreach \x/\lbl in {0/0.0, 1/0.3, 2/0.7, 3/1.0}
    \node[font=\tiny] at ({\x*\s+\s/2},{4*\s+0.1}) {\lbl};
  \node[font=\scriptsize] at ({2*\s},{4*\s+0.45}) {$\omega$};
  \foreach \y/\lbl in {3/0.0, 2/0.3, 1/0.7, 0/1.0}
    \node[font=\tiny, anchor=east] at (-0.1,{\y*\s+\s/2}) {\lbl};
  \node[font=\scriptsize, rotate=90, anchor=south] at (-0.5,{2*\s}) {$\tau$};
  % tau=0.0
  \hc{0}{3}{31.9}{62}{white}\hc{1}{3}{34.6}{72}{white}\hc{2}{3}{32.4}{64}{white}\hc{3}{3}{36.0}{77}{white}
  % tau=0.3
  \hc{0}{2}{36.5}{78}{white}\hc{1}{2}{42.5}{100}{white}\hc{2}{2}{35.7}{76}{white}\hc{3}{2}{\textbf{26.6}}{43}{black}
  % tau=0.7
  \hc{0}{1}{32.0}{62}{white}\hc{1}{1}{29.4}{53}{white}\hc{2}{1}{34.2}{70}{white}\hc{3}{1}{34.4}{71}{white}
  % tau=1.0
  \hc{0}{0}{37.5}{82}{white}\hc{1}{0}{32.7}{65}{white}\hc{2}{0}{34.0}{70}{white}\hc{3}{0}{38.5}{86}{white}
\end{scope}
% --- Panel (c): 70B, C = 128 ---
\begin{scope}[xshift=10.2cm]
  \node[font=\footnotesize\bfseries, anchor=south] at ({2*\s},{4*\s+0.25}) {(c) 70B, $C = 128$};
  \foreach \x/\lbl in {0/0.0, 1/0.3, 2/0.7, 3/1.0}
    \node[font=\tiny] at ({\x*\s+\s/2},{4*\s+0.1}) {\lbl};
  \node[font=\scriptsize] at ({2*\s},{4*\s+0.45}) {$\omega$};
  \foreach \y/\lbl in {3/0.0, 2/0.3, 1/0.7, 0/1.0}
    \node[font=\tiny, anchor=east] at (-0.1,{\y*\s+\s/2}) {\lbl};
  \node[font=\scriptsize, rotate=90, anchor=south] at (-0.5,{2*\s}) {$\tau$};
  % tau=0.0
  \hc{0}{3}{19.6}{18}{black}\hc{1}{3}{18.4}{14}{black}\hc{2}{3}{25.4}{39}{black}\hc{3}{3}{19.1}{16}{black}
  % tau=0.3
  \hc{0}{2}{27.5}{46}{black}\hc{1}{2}{20.6}{22}{black}\hc{2}{2}{\textbf{14.6}}{0}{black}\hc{3}{2}{15.6}{4}{black}
  % tau=0.7
  \hc{0}{1}{28.2}{49}{black}\hc{1}{1}{26.4}{42}{black}\hc{2}{1}{24.6}{36}{black}\hc{3}{1}{20.5}{21}{black}
  % tau=1.0
  \hc{0}{0}{21.2}{24}{black}\hc{1}{0}{25.7}{40}{black}\hc{2}{0}{27.2}{45}{black}\hc{3}{0}{15.4}{3}{black}
\end{scope}
\end{tikzpicture}
\caption{$\widehat{\text{PoA}}$ heatmaps for the $4 \times 4$ parameter sweep, shown on a \emph{shared} color scale (PoA 14.6 $\to$ 42.5 spans the full intensity range). (a)~340B below saturation ($C = 64$): all 16 configurations yield $\widehat{\text{PoA}} = 18.7 \pm 0.10$, appearing uniformly light. (b)~340B at saturation ($C = 128$): $\widehat{\text{PoA}} = 34.3 \pm 3.7$, almost every cell deeply tinted ($1.6\times$ spread; moderate, unstructured variation). (c)~70B 1P/2D at saturation ($C = 128$): $\widehat{\text{PoA}}$ ranges from 14.6 to 28.2 ($1.9\times$ spread)---visibly rugged within the panel yet still cooler than (b) on the shared scale, showing that the 70B 1P/2D at saturation is more parameter-sensitive but absolutely more efficient than the 340B. Bold values mark minima.}
\label{fig:pareto-heatmap}
\end{figure}

\subsection{Experiment 3: Adaptive versus Static Routing}
\label{sec:res-adaptive}

Experiment~3 tests whether the adaptive controller can exploit regime-dependent routing under a realistic load spike ($C = 32 \to 128 \to 32$, 120/180/120\,s, 3 iterations per strategy).
We report the 340B (1P/2D) and 70B (1P/5D) configurations in detail.
The 70B 1P/2D configuration was also tested ($n=3$): the adaptive controller reduced PoA from $23.1 \pm 1.6$ to $10.7 \pm 0.5$ ($2.2\times$) and TTFT P99 from $25.9 \pm 2.1$\,s to $3.4 \pm 0.4$\,s ($7.6\times$), the strongest TTFT improvement across all configurations, achieved with the same $(\tau=0.7, \omega=1.0)$ parameter setting.
Switch latency was $48.6 \pm 13.0$\,s.
Two strategies are compared:
\begin{itemize}[nosep]
  \item \textbf{Static}: default parameters ($\tau = 0$, $\omega = 1$) throughout all phases.
  \item \textbf{Adaptive}: the SaturationDetector monitors TTFT P99 via 5\,s Prometheus polling with EWMA; upon detecting \textsc{transition}, a zero-downtime port switch redirects traffic from a default-parameter frontend (port~8000) to a pre-started optimal-parameter frontend (port~8001, $\tau = 0.7$, $\omega = 1.0$). No frontend restart is required.
\end{itemize}

\paragraph{Results.}
Tables~\ref{tab:adaptive-340b} and~\ref{tab:adaptive-70b} summarize the per-phase metrics for both models.

\begin{table}[t]
\centering
\small
\caption{Experiment~3: Adaptive vs.\ static routing, Nemotron-4-340B (1P/2D, mean $\pm$ sample std across 3 iterations).}
\label{tab:adaptive-340b}
\begin{tabular}{llrrrr}
\toprule
\headrow \textbf{Strategy} & \textbf{Phase} & \textbf{$\widehat{\text{PoA}}$} & \textbf{TTFT P99 (s)} & \textbf{ITL P99 (ms)} & \textbf{rps} \\
\midrule
Static   & Below     & $19.51 \pm 0.07$ & $0.29$ & $22.3$ & 5.5 \\
Static   & Saturated & $33.15 \pm 1.14$ & $28.26 \pm 5.9$ & $22.7$ & 18.2 \\
Static   & Recovery  & $19.52 \pm 0.04$ & $0.30$ & $23.8$ & 5.5 \\
\midrule
Adaptive & Below     & $19.55 \pm 0.02$ & $0.29$ & $22.5$ & 5.5 \\
Adaptive & Saturated & $\mathbf{25.78 \pm 0.07}$ & $\mathbf{5.92 \pm 0.06}$ & $22.4$ & 11.6 \\
Adaptive & Recovery  & $19.55 \pm 0.08$ & $0.57$ & $22.8$ & 5.5 \\
\bottomrule
\end{tabular}
\end{table}

\begin{table}[t]
\centering
\small
\caption{Experiment~3: Adaptive vs.\ static routing, Llama-3.1-70B (1P/5D, mean $\pm$ sample std across 3 iterations).}
\label{tab:adaptive-70b}
\begin{tabular}{llrrrr}
\toprule
\headrow \textbf{Strategy} & \textbf{Phase} & \textbf{$\widehat{\text{PoA}}$} & \textbf{TTFT P99 (s)} & \textbf{ITL P99 (ms)} & \textbf{rps} \\
\midrule
Static   & Below     & $14.72 \pm 0.04$ & $0.11$\textsuperscript{\dag} ($n{=}2$) & $9.9$ & 16.7 \\
Static   & Saturated & $66.42 \pm 12.2$ & $8.19 \pm 5.6$ & $11.5$ & 50.9 \\
Static   & Recovery  & $14.68 \pm 0.04$ & $0.094$ & $9.9$ & 16.8 \\
\midrule
Adaptive & Below     & $14.66 \pm 0.03$ & $0.106$ & $10.0$ & 16.8 \\
Adaptive & Saturated & $\mathbf{21.53 \pm 0.17}$ & $\mathbf{4.22 \pm 0.05}$ & $11.3$ & 44.3 \\
Adaptive & Recovery  & $14.75 \pm 0.11$ & $0.128$ & $9.9$ & 16.8 \\
\bottomrule
\end{tabular}
\end{table}

Figure~\ref{fig:adaptive-poa} visualizes the PoA comparison across phases for both models.

\begin{figure}[t]
\centering
\begin{tikzpicture}
\begin{axis}[
  ybar,
  bar width=10pt,
  width=0.95\columnwidth,
  height=5.5cm,
  ylabel={$\widehat{\text{PoA}}$},
  symbolic x coords={Below,Saturated,Recovery},
  xtick=data,
  ymin=0, ymax=70,
  ytick={0,10,20,30,40,50,60,70},
  legend style={at={(0.98,0.98)}, anchor=north east, font=\scriptsize},
  nodes near coords,
  nodes near coords style={font=\tiny, above},
  every node near coord/.append style={/pgf/number format/fixed, /pgf/number format/precision=1},
  enlarge x limits=0.3,
  grid=major,
  grid style={gray!30},
]
\addplot[fill=gOne!55!white, draw=gOne, line width=0.4pt] coordinates {(Below,19.51) (Saturated,33.15) (Recovery,19.52)};
\addplot[fill=gOne!20!white, draw=gOne, line width=0.4pt, postaction={pattern=north east lines, pattern color=gOne!70}] coordinates {(Below,19.55) (Saturated,25.78) (Recovery,19.55)};
\addplot[fill=gThree!55!white, draw=gThree, line width=0.4pt] coordinates {(Below,14.72) (Saturated,66.42) (Recovery,14.68)};
\addplot[fill=gThree!20!white, draw=gThree, line width=0.4pt, postaction={pattern=north east lines, pattern color=gThree!70}] coordinates {(Below,14.66) (Saturated,21.53) (Recovery,14.75)};
\legend{340B Static, 340B Adaptive, 70B Static, 70B Adaptive}
\end{axis}
\end{tikzpicture}
\caption{PoA comparison under load spike (mean of 3 iterations each). Below saturation and in recovery, both strategies produce identical PoA on each model. At saturation, the adaptive controller reduces PoA on the 340B ($33.15 \to 25.78$, 22\%) and on the 70B ($66.42 \to 21.53$, $3.1\times$). The 70B benefits more from parameter switching because its 5-worker routing game ($m = 5$) has a larger action space than the 340B's 2-worker game ($m = 2$).}
\label{fig:adaptive-poa}
\end{figure}
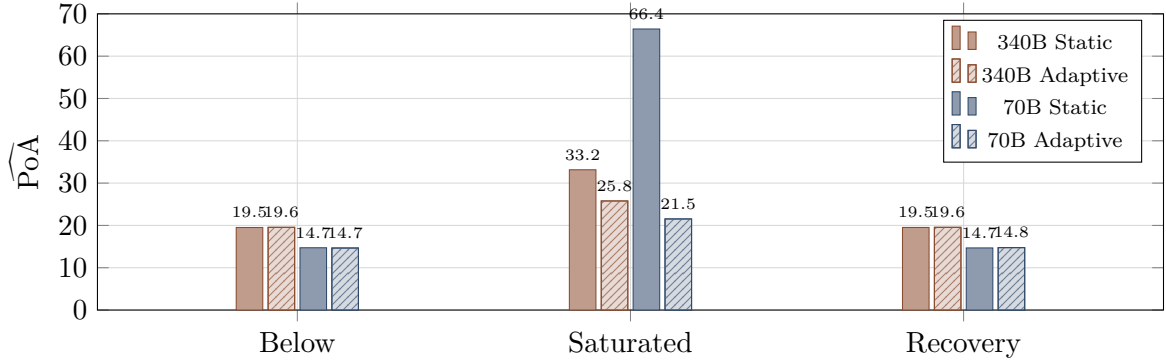

\paragraph{340B results.}
On the 340B, the adaptive controller reduces aggregate saturated-phase TTFT P99 by $\mathbf{4.8\times}$ ($28.26 \pm 5.9$\,s $\to$ $5.92 \pm 0.06$\,s) and $\widehat{\text{PoA}}$ by 22\% ($33.15 \to 25.78$).
The zero-downtime switch fires at $54.8 \pm 0.2$\,s into the 180\,s saturated phase, remarkably consistent across all 3 iterations, leaving $\sim$125\,s under the optimal frontend.
The aggregate number therefore mixes two regimes and understates the steady-state benefit; we report both:

\begin{itemize}[nosep]
  \item \textbf{Saturated-phase aggregate (headline):} TTFT P99 $28.26 \to 5.92$\,s, a $4.8\times$ reduction.
  \item \textbf{Post-switch steady state:} TTFT P99 $\sim$0.97\,s (versus the static $\sim$28.3\,s), a $\sim$29$\times$ reduction.
  \item \textbf{Pre-switch ($\sim$55\,s at default params):} TTFT P99 tracks the static baseline.
\end{itemize}

The aggregate understates the steady-state benefit by a factor of $\sim$6, while conversely the $\widehat{\text{PoA}}$ reduction (22\%) may partially reflect the transient (see saturated-phase discussion in Section~\ref{sec:disc-limitations}). Reported throughput costs refer to the saturated phase only: the adaptive strategy achieves 11.6\,rps vs.\ 18.2\,rps for static (a 36\% reduction), reflecting a shifted operating point that prioritizes latency over throughput at saturation.

The $\widehat{\text{PoA}}$ improvement is less dramatic than TTFT because with only 2 decode workers, the routing game is inherently constrained, and parameter tuning cannot differentiate outcomes as much as with more workers.
Even so, the TTFT improvement shows that moderate routing-efficiency gains translate to large user-facing latency benefits at saturation.

\paragraph{70B results.}
On the 70B 1P/5D ($m = 5$), the adaptive strategy reduces $\widehat{\text{PoA}}$ from $66.42 \pm 12.2$ to $\mathbf{21.53 \pm 0.17}$, a $\mathbf{3.1\times}$ reduction and our strongest headline improvement.
This is the statistically robust metric for this configuration: the static $\widehat{\text{PoA}}$ baseline has CV $= 18\%$ and the $3.1\times$ ratio does not straddle zero in any credible interval.
Mean TTFT P99 also drops ($8.19 \to 4.22$\,s, a $1.94\times$ reduction), but we demote this number in the headline framing: the static baseline individual iterations are 2.6\,s, 13.8\,s, 8.2\,s (CV $= 68\%$), and the $n{=}3$ 95\% $t$-CI spans $[-5.7, 22.1]$\,s, so the ratio is dominated by the middle iteration and an $n \geq 5$ rerun would be needed to report it with confidence.\footnote{One iteration of the 70B static below-saturation phase produced an anomalous TTFT P99 of 2.84\,s (vs.\ $\sim$0.1\,s in the other two), likely due to insufficient cool-down from the previous iteration's saturated phase. The below-saturation TTFT (\dag) in Table~\ref{tab:adaptive-70b} reports the mean of the two non-anomalous iterations (0.127\,s, 0.089\,s); including the anomalous iteration, the 3-iteration mean would be $1.02 \pm 1.58$\,s, dominated by the outlier. For symmetry, the corresponding 3-iteration $\widehat{\text{PoA}}$ under the full 3-iteration treatment is $14.73 \pm 0.06$ (unaffected by the TTFT outlier).}
The switch fired at $58.3 \pm 15.7$\,s into the saturated phase.
The larger routing game enables more differentiation between parameter configurations, yielding the stronger $\widehat{\text{PoA}}$ reduction.

\paragraph{Stability.}
On both models, the adaptive strategy has dramatically lower variance than static:
\begin{itemize}[nosep]
  \item 340B: $\sigma_{\text{PoA}} = 0.07$ (adaptive) vs.\ 1.14 (static); $\sigma_{\text{TTFT}} = 0.06$\,s vs.\ 5.9\,s.
  \item 70B: $\sigma_{\text{PoA}} = 0.17$ vs.\ 12.2; TTFT consistently $\sim$4.2\,s vs.\ highly variable static ($8.2 \pm 5.6$\,s).
\end{itemize}
The static strategy's high TTFT variance (21.5--31.8\,s on the 340B) reflects stochastic queue dynamics at the saturation knee. The adaptive controller eliminates this variance by preventing queue explosion via the optimal frontend.
In recovery, both strategies return to the below-saturation baseline on both models, confirming clean regime transitions.

\section{Discussion}
\label{sec:discussion}

\subsection{When Game Theory Applies}
\label{sec:disc-when}

Three assumptions of classical game theory are imperfect fits for LLM inference:

\paragraph{Rationality and the mechanism design interpretation.}
Inference requests are passive workloads, not strategic agents.
The mechanism-design framing in Section~\ref{sec:introduction} side-steps this: PoA measures the efficiency of the router-as-mechanism, requiring no rationality assumption on requests.
Strategic considerations re-enter at the \emph{tenant} level in multi-tenant clusters, where users have genuine incentives to game the allocator (Section~\ref{sec:disc-implications}).

\paragraph{Complete information.}
Output sequence lengths are unknown until generation completes, future arrival patterns are stochastic, and execution times depend on dynamic batch composition and memory pressure.
The games we formalize are therefore games of incomplete information.
Our controller addresses this pragmatically: rather than computing equilibria under uncertainty, it detects regime shifts from observable metrics and applies pre-computed parameter settings.

\paragraph{Static analysis.}
Classical game theory analyzes one-shot interactions, but inference serving is a continuous system where GPU memory, batch composition, and queue depths evolve at millisecond timescales.
Dynamic game models (differential games, repeated games with discounting) are the appropriate generalization.
Our saturation analysis (Section~\ref{sec:saturation}) takes a step in this direction by characterizing how the game's structure changes with load, but a full dynamic treatment remains future work.

\subsection{Scope and Boundaries}
\label{sec:disc-limitations}

This work studies two dense models on a 3-node cluster under a controlled workload.
Several limitations bound the generality of our findings.

\paragraph{Workload homogeneity.}
All experiments use a single workload profile: 5 prompt templates, 128 input tokens, 256 max output tokens, and deterministic generation (temperature~0.0).
With only 5 distinct prefixes and $m = 2$--$5$ workers, the KV cache prefix space is trivially partitionable, making the cache placement game (Game~2) degenerate: the parameter invariance observed below saturation in the Pareto sweep may partly reflect this trivial cache problem rather than a general property.
Furthermore, with 128-token inputs, KV cache blocks are small relative to the 192\,GB HBM3e capacity, so cache tier spillover (Game~2's multi-tier cost structure) is never exercised; all blocks remain in HBM throughout.
Deterministic generation (temperature~0.0) eliminates output-length variance that would create heterogeneous decode loads (Game~3) and limits the diversity of congestion patterns.
Whether the three-regime PoA structure and parameter invariance hold under variable-length, multi-turn, or production-trace workloads is an open question.
A \texttt{long\_context} workload type was implemented in the experiment code but not exercised in the reported experiments.

\paragraph{Game 1 (P/D allocation) is not empirically validated.}
All experiments use fixed prefill/decode splits; the Planner's dynamic GPU reallocation (Game~1) is analyzed theoretically (Proposition~\ref{prop:pd-equilibrium}) but never tested.
GPU utilization metrics for the P/D equilibrium analysis are synthetically generated from the fixed topology.
This means our empirical findings validate Games~2 and~3 (cache placement and routing) but not Game~1.

\paragraph{Scope.}
The regime structure holds across both models and all tested topologies; extending to MoE architectures, larger clusters ($m \gg 5$), and heterogeneous production traffic would map the boundaries of where these properties hold.

\paragraph{Routing baselines: PoA captures temporal dynamics.}
Our PoA measurements characterize \dynamo's built-in KV-aware greedy router.
A static counterfactual analysis (assigning requests to workers under the same cost model (Section~\ref{sec:ctrl-implementation}) using round-robin, random uniform, and power-of-two-choices policies) shows all policies within $\sim$0.3--10\% of the Hungarian optimal at $C \geq 8$ for both $m = 2$ and $m = 5$ worker topologies (at $C < 8$, the small number of in-flight requests causes wider divergence between policies).\footnote{At a few concurrency levels on $m=5$ workers, the KV-aware greedy policy achieves a PoA ratio slightly below 1.0 (e.g., 0.968 at $C=64$). This occurs because the frozen-latency Hungarian assignment is not a strict lower bound when the greedy policy exploits KV cache overlap information that the Hungarian cost matrix approximates imperfectly.}
This is a key result: it confirms that the empirically measured $\widehat{\text{PoA}}$ of ${\sim}19\times$ (340B) and ${\sim}7\times$ (70B) is driven by \emph{temporal} dynamics (queuing delays, prefill/decode contention, and batch scheduling), not by the static assignment algorithm.
The two-order-of-magnitude gap between the static routing $\widehat{\text{PoA}}$ ($\sim$1.02--1.08) and the empirical $\widehat{\text{PoA}}$ ($\sim$7--19) demonstrates that the $\widehat{\text{PoA}}$ metric captures system-level inefficiency that conventional routing comparisons miss.
Deriving queuing-theoretic baselines (e.g., M/G/$k$ predictions of the saturation knee) remains a complementary direction.

Game~2 (KV cache placement) is validated indirectly through the $\omega$-dimension of the Pareto sweep (Section~\ref{sec:res-pareto-saturated}): the cache game's PoA contribution is negligible below saturation and creates a measurable congestion-affinity tradeoff at saturation.
The $\omega$ sweep was bounded at $\omega = 1.0$ (\dynamo's default upper end); the $\omega > 1$ regime---cache affinity weighted more heavily than active-block load---was not exercised and could alter the at-saturation tradeoff between cache-warm concentration and load imbalance.
Isolating per-block cache placement decisions (tier promotion, eviction under capacity pressure) would require controlled cache-capacity experiments with diverse prefix-sharing workloads.

The adaptive controller uses fixed thresholds and fixed per-regime parameters.
Continuous parameter interpolation or online learning would be natural extensions; we use the simplest viable controller to isolate the value of regime detection.
Note that the adaptive controller's TTFT improvement on the 340B comes at a throughput cost: 11.6\,rps (adaptive) vs.\ 18.2\,rps (static), a 36\% reduction.
This reflects the shifted operating point: the optimal-parameter frontend prioritizes latency over throughput at saturation.
Whether this tradeoff is desirable depends on the SLO; a production deployment would need to select the target point on the Pareto frontier explicitly.

\subsection{Implications for System Design}
\label{sec:disc-implications}

Three design implications follow from the regime structure:

\paragraph{Expose the PoA as a first-class metric.}
A rising PoA warns that the system is entering a regime where selfish routing degrades rapidly---even when absolute latency is still within SLO bounds.
Serving frameworks should export PoA estimates alongside TTFT and ITL.

\paragraph{Design for regime transitions, not steady state.}
The system spends most of its time in transient states between regimes, not at equilibrium.
Planners and routers should be evaluated on their behavior during transitions (convergence speed, overshoot, oscillation), not just steady-state throughput.

\paragraph{Consider the multi-tenant game.}
In shared clusters, tenants have genuine strategic incentives that single-tenant benchmarks cannot capture.
Game-theoretic mechanism design (strategy-proof auctions, incentive-compatible allocation) becomes directly relevant for multi-tenant disaggregated inference, where tenants might overprovision requests, send dummy prefills to warm cache, or time batch jobs to avoid contention.
\dynamo's Planner currently optimizes aggregate SLOs; a multi-tenant game-theoretic extension would surface incentive properties under strategic tenant behavior and characterize the fairness guarantees it provides---a natural next step jointly with the Planner team.

\section{Conclusion}
\label{sec:conclusion}

This paper modeled NVIDIA \dynamo's disaggregated inference architecture as three coupled games and measured the Price of Anarchy of one of them---request routing---on a 3-node B200 cluster across two models and three topologies.

The empirical observation was the flatness.
Router parameters ($\tau$, $\omega$) did not measurably affect $\widehat{\text{PoA}}$ below saturation.
The index varies by less than $\pm 0.10$ across 16 $(\tau, \omega)$ configurations on the 340B, $\pm 0.08$ on the 70B 1P/2D, and $\pm 0.06$ on the 70B 1P/5D (single trial per configuration; cross-config spread, not measurement uncertainty---see the §\ref{sec:results} measurement note).
Parameters start mattering only at the first post-knee grid point, and that point sits at the same concurrency ($C = 128$) on a 70B-parameter and a 340B-parameter model despite a $4.9\times$ size gap and different TP configurations and KV transfer paths.
One plausible mechanical explanation: both topologies use a \emph{single} prefill worker on identical B200 hardware with identical 128-token inputs, so single-prefill-worker compute exhaustion is expected to occur at a similar in-flight-request count.
Our grid cannot resolve whether the true knees are exactly co-located within $(96, 128]$; denser sweeps on other hardware would test whether the coincidence is deeper than single-worker prefill saturation.

The adaptive controller exploits this structure with the simplest possible mechanism: detect the regime, switch the parameters.
Measured improvements scale with routing-game size: the largest $\widehat{\text{PoA}}$ reduction lands on the 70B 1P/5D configuration ($m=5$), with smaller but consistent $\widehat{\text{PoA}}$ and order-of-magnitude TTFT P99 reductions on the 70B and 340B 1P/2D topologies (Section~\ref{sec:res-adaptive}).
Run-to-run variance shrinks by an order of magnitude in every case.
The deep-saturation $\widehat{\text{PoA}}$ blow-up ($\sim$200--309) is consistent with single-prefill-worker compute exhaustion rather than routing suboptimality; attributing this specifically to Game~1 (P/D split) resource allocation would require varying the P/D split (e.g., 2P/1D), which our fixed-topology experiments cannot do.

Our framework confirms~\citep{nisan2007, roughgarden2015} that game theory's value for inference systems is \emph{analytical, not algorithmic}: vocabulary to characterize behavior, metrics to quantify suboptimality, and regime analysis to guide operations---without runtime equilibrium computation.

\subsection{Future Work}

Several directions extend this work:

\paragraph{Scaling and queuing baselines.}
Parameter sensitivity at saturation differs by model ($1.6\times$ spread on the 340B vs.\ $1.9\text{--}2.0\times$ on the 70B), with the 70B showing similar sensitivity on both $m=2$ and $m=5$ topologies.
Production-scale clusters ($m \gg 5$) would reveal whether this trend continues and whether PoA invariance below saturation eventually breaks.
Static counterfactual analysis already shows all routing policies (round-robin, random, power-of-two-choices) within $\sim$0.3--10\% of the Hungarian optimal at $C \geq 8$, confirming the measured PoA is driven by temporal dynamics rather than algorithm choice (Section~\ref{sec:disc-limitations}).
A queuing-theoretic baseline (e.g., M/G/$k$) would identify which aspects of the regime structure are already predicted by classical models and where the game-theoretic lens adds new insight.

\paragraph{Caching, workloads, and dynamics.}
Our Pareto sweep validates Game~2's aggregate effect (Section~\ref{sec:res-pareto-saturated}); isolating per-block placement decisions requires controlled cache-capacity experiments with diverse prefix-sharing workloads.
Heterogeneous and production-trace workloads would test the generality of the regime structure.
Dynamic game models (differential games, repeated games with discounting) would extend the analysis from stationary equilibria to the transient dynamics that dominate bursty production traffic.

\paragraph{Architecture extensions.}
MoE models introduce a nested congestion game (requests route to GPUs (outer), then to experts (inner)) where the PoA may compound across layers.
In multi-tenant clusters, tenants are genuine strategic agents; mechanism design~\citep{mahajan2020} for strategy-proofness becomes directly relevant (Section~\ref{sec:disc-implications}).

\paragraph{P/D split adaptation.}
A game-theory-informed Planner could use the PoA signal to drive faster, larger P/D adjustments than the load-based Planner's $\pm 1$ worker per 30-second interval (the SLA-based Planner, production default since Dynamo 0.4, computes target replica counts directly but still runs on minute-scale cadence), directly attacking the saturation-regime bottleneck identified above.

\appendix

\section{Game Theory Definitions}
\label{sec:appendix-defs}

We state the standard game-theoretic definitions used throughout the paper; see~\citet{nisan2007} for a comprehensive treatment.

\begin{definition}[Normal-form game]
A \emph{normal-form game} $\Gamma = (N, (S_i)_{i \in N}, (u_i)_{i \in N})$ consists of a set of players $N = \{1, \ldots, n\}$, a strategy set $S_i$ for each player $i$, and a utility function $u_i : \prod_{j \in N} S_j \to \R$ for each player $i$.
\end{definition}

\begin{definition}[Nash equilibrium]
A strategy profile $s^* = (s_1^*, \ldots, s_n^*)$ is a \emph{Nash equilibrium} if no player can improve their utility by unilaterally deviating:
\[
  u_i(s_i^*, s_{-i}^*) \geq u_i(s_i, s_{-i}^*) \quad \forall i \in N, \; \forall s_i \in S_i
\]
where $s_{-i}^*$ denotes the strategies of all players except $i$.
\end{definition}

\begin{definition}[Congestion game~\citep{rosenthal1973}]
A \emph{congestion game} consists of a set of players $N$, a set of resources $R$, and for each player $i$, a collection of feasible resource subsets $\Sigma_i \subseteq 2^R$.
Each resource $r \in R$ has a cost function $c_r : \N \to \R$ that depends only on the number of players using $r$.
The cost to player $i$ under strategy profile $\sigma$ is $C_i(\sigma) = \sum_{r \in \sigma_i} c_r(n_r(\sigma))$, where $n_r(\sigma)$ is the number of players using resource $r$.
Every congestion game is an \emph{exact potential game}~\citep{monderer1996}: there exists a potential function $\Phi(\sigma) = \sum_{r \in R} \sum_{k=1}^{n_r(\sigma)} c_r(k)$ such that any unilateral deviation that decreases a player's cost also decreases $\Phi$.
Consequently, every congestion game possesses at least one pure Nash equilibrium.
\end{definition}

\begin{definition}[Price of Anarchy~\citep{roughgarden2002}]
The \emph{Price of Anarchy} of a game is the ratio of the social cost at the worst Nash equilibrium to the social cost at the social optimum:
\[
  \text{PoA} = \frac{\max_{s^* \in \text{NE}} \text{SC}(s^*)}{\min_{s \in S} \text{SC}(s)}
\]
where $\text{SC}(s) = \sum_{i} C_i(s)$ is the social cost (sum of individual costs) and $\text{NE}$ is the set of Nash equilibria.
A Price of Anarchy of 1 means selfish behavior is socially optimal; larger values indicate inefficiency due to uncoordinated decision-making.
\end{definition}

\begin{definition}[Pareto optimality]
A strategy profile $s$ is \emph{Pareto optimal} (or Pareto efficient) if there is no alternative profile $s'$ such that $u_i(s') \geq u_i(s)$ for all players $i$ and $u_j(s') > u_j(s)$ for at least one player $j$.
The \emph{Pareto frontier} is the set of all Pareto-optimal outcomes.
In the context of LLM serving, the Pareto frontier typically characterizes the tradeoff surface across throughput, TTFT, ITL, and cost.
\end{definition}

\begin{definition}[Wardrop equilibrium~\citep{wardrop1952}]
In a routing game with a continuum of infinitesimally small players, a \emph{Wardrop equilibrium} is a flow assignment where all used paths between any origin-destination pair have equal cost, and no unused path has lower cost.
This is the continuous-flow analog of Nash equilibrium, applicable when the number of requests is large relative to the number of GPUs.
\end{definition}

\ifanon
\else
  \section*{Acknowledgements}
  We thank Kevin Cochrane for the initial discussions that motivated this line of inquiry, and Vultr for providing the compute infrastructure used in our experiments.
\fi

\bibliography{references}

\end{document}